\documentclass[acmsmall,screen]{acmart}
\usepackage{graphicx} 

\usepackage{gdefs}
\usepackage{refs}
\usepackage{mathtools}
\usepackage{amsmath}

\usepackage{kbordermatrix}
\usepackage{cancel}
\DeclareMathOperator{\tensor}{\otimes}     
\DeclareMathOperator{\bigO}{\mathcal{O}}   
\DeclareMathOperator{\bigOmega}{\Omega}    
\usepackage[export]{adjustbox}
\usepackage{multirow}
\usepackage{float}
\usepackage[justification=centering]{caption}
\usepackage[justification=centering]{subcaption}
\usepackage{amsfonts}
\usepackage{amsthm}
\usepackage{stmaryrd}
\usepackage{autobreak}
\usepackage{enumitem}
\setlist{noitemsep,topsep=0.7mm,leftmargin=*}
\usepackage{setspace}
\usepackage{hyperref}
\usepackage{graphbox}
\usepackage[framemethod=default]{mdframed}
\usepackage[ruled,linesnumbered]{algorithm2e}
\usepackage{nicematrix}
\SetKwInput{Input}{Input}
\SetKwInput{Output}{Output}
\SetKwBlock{Begin}{begin}{end}

\newtheorem{Construction}{\indent\,Construction}
\usepackage{xspace}
\usepackage{csquotes}
\usepackage{wrapfig}
\usepackage{rotating}
\usepackage{tablefootnote}
\usepackage{nccmath}

\makeatletter
\def\thmhead@plain#1#2#3{%
  \thmname{#1}\thmnumber{\@ifnotempty{#1}{ }\@upn{#2}}%
  \thmnote{ {\the\thm@notefont#3}}}
\let\thmhead\thmhead@plain
\makeatother
\theoremstyle{definition}
\newtheorem*{theorem*}{Theorem}

\newcounter{LineNumber}

\usepackage[color]{changebar}
\cbcolor{red}

\newcommand{\EQ}{{\textit{EQ}}}

\newcommand{\GHZ}{{\textit{GHZ}}\xspace}
\newcommand{\BV}{{\textit{BV}}\xspace}

\newcommand{\HadamardFamily}{{\mathcal{H}}}

\newcommand{\Omit}[1]{}

\definecolor{lightgray}{rgb}{0.55,0.52,0.54}

\newcommand{\OnlySupplemental}[1]{}    

\renewcommand{\subsubsection}[1]{\medskip\noindent{\textbf{#1}}}
\newcommand{\TIDD}{\textrm{TIDD}\xspace}
\newcommand{\TIDDs}{\textrm{TIDDs}\xspace}
\date{}

\keywords{Decision diagram, nested-word automata, quantum simulation}

\begin{document}

\title{Do CFLOBDDs Actually Make Use of Linear Structure?}

\author{Meghana Sistla}
\orcid{0000-0002-4215-0651}
\affiliation{%
  \institution{University of Texas at Austin}
  \city{Austin}
  \country{USA}
}
\email{mesistla@utexas.edu}

\author{Swarat Chaudhuri}
\orcid{0000-0002-6859-1391}
\affiliation{%
  \institution{University of Texas at Austin}
  \city{Austin}
  \country{USA}
}
\email{swarat@cs.utexas.edu}

\author{Thomas Reps}
\orcid{0000-0002-5676-9949}
\affiliation{%
  \institution{University of Wisconsin}
  \city{Madison}
  \country{USA}
}
\email{reps@cs.wisc.edu}

\begin{abstract}
Binary Decision Diagrams (BDDs) are a widely used data structure for efficient Boolean function representation. Context-Free-Language Ordered Binary Decision Diagrams (CFLOBDDs) are a recently introduced hierarchical data structure that
can, in the best case, exhibit
exponential compression over BDDs and double-exponential compression over decision trees.
Roughly speaking, a CFLOBDD is a finite, acyclic, non-recursive hierarchical finite-state machine (HFSM) (with some additional restrictions).
In this paper, we investigate the role of \emph{linear structure} in CFLOBDDs---a property that connects them to Nested-Word Automata (NWAs) and Visibly Pushdown Automata (VPAs)---and examine whether CFLOBDDs actively exploit this structure beyond their well-studied hierarchical properties. We demonstrate that linear structure, in conjunction with hierarchical structure, plays a crucial role in enabling CFLOBDDs to achieve efficient function compression. Furthermore, we show that removing linearity from CFLOBDDs leads to a significant blowup in representation size, resulting in degraded performance in the domain of quantum-circuit simulation.
\end{abstract}

\maketitle

\section{Introduction}
\label{Se:ch6_intro}


Efficient representation of data encoded as Boolean functions---used to model matrices, relations, etc.---is fundamental to several areas of computer science. Binary Decision Diagrams \cite{toc:Bryant86,brace1991efficient} (BDDs) provide a compact representation of such functions and have become a cornerstone of symbolic methods in verification, program analysis, and, more recently,
simulation of quantum circuits on classical hardware.
More recently, a new data structure, called \emph{Context-Free-Language Ordered Binary Decision Diagrams} (CFLOBDDs)~\cite{TOPLAS:SCR24}, was introduced,
which introduces hierarchy into the decision-diagram representation, and can, in the best case, exhibit exponential compression over the BDD representation of a function, and double-exponential compression over the decision-tree representation of a function.

CFLOBDDs are a kind of 
finite, acyclic, non-recursive
hierarchical finite-state machine (HFSM) \cite{DBLP:journals/toplas/AlurBEGRY05}, of a special form:
every grouping at level $l$ in a CFLOBDD has a ``call'' to another grouping at level $l-1$ through an AConnection, \textit{followed} by calls to groupings at level $l-1$ through BConnections.
The number of BConnections is determined by the number of exits of the AConnection.
A CFLOBDD has a hierarchical structure through these ``calls'' to lower-level groupings;
in contrast, BDDs have no analogue of such ``calls,'' and thus have nothing that corresponds to a call hierarchy.

Appendix K of~\cite{2211.06818} discusses how CFLOBDDs can be viewed as going somewhat beyond HFSMs, in that they have some of the features of Nested-Word Automata (NWAs) (or Visibly Pushdown Automata (VPAs)).
In particular, like NWAs, CFLOBDDs combine a hierarchical structure with a \emph{linear structure}.
\sectref{ch6_NWADefinition} formally defines the linear structure of NWAs (similar to Appendix K of~\cite{2211.06818}) and shows how CFLOBDDs are a special case of NWAs.

\cite{TOPLAS:SCR24} shows how the
call hierarchy in CFLOBDDs can provide an exponential reduction in size, compared with BDDs, both 
theoretically and in practice.
In this paper, we explore the question of whether
some of the benefits enjoyed by CFLOBDDs are due to
the linear structure.
To study this question, we introduce a new kind decision diagram, called Tree-automata Inspired Decision Diagrams (TIDDs), that is similar to CFLOBDDs, but without the linear structure: TIDDs possess only the hierarchical property of CFLOBDDs (and can be viewed as a restricted form of deterministic tree automata). 

\sectrefs{ch6_separation_intuition}{ch6_separation_example} discuss how 
CFLOBDDs can represent functions better than \TIDDs because of the lack of linear structure in \TIDDs, and present an example of a function $f$ for which there is an exponential gap between CFLOBDDs and \TIDDs in terms of representation size---i.e., a CFLOBDD for $f$ can be exponentially smaller than a \TIDD for $f$ with the same variable ordering. 

\paragraph{Organization.}
The remainder of the paper covers the following material:
    \sectref{BackgroundAndOverview} gives background about how CFLOBDDs are a restricted form of NWAs (and how CFLOBDDs incorporate both hierarchical and linear structure), and gives an overview of how we can create a hierarchical decision diagram that lacks linear structure.
    \sectref{ch6_tidds} defines
    \TIDDs, a new data structure to represent Boolean functions, relations, etc.
    \sectref{ch6_ops} gives an overview of the operations that \TIDDs support.
    To explore the relative expressivity of CFLOBDDs and TIDDs, \sectref{ch6_examples} discusses examples of three
    kinds: (i) ones for which \TIDDs are as good as CFLOBDDs and exponentially better than BDDs (\sectref{ch6_efficient_relations}), 
    (ii) ones that provide intuition about how CFLOBDDs use the linear structure and can represent functions better than \TIDDs (\sectref{ch6_separation_intuition}),
    and
    (iii) an example that illustrates
    an exponential gap between CFLOBDDs and \TIDDs (with same variable ordering) (\sectref{ch6_separation_example}).
    \sectref{ch6_eval} presents
    some sample evaluations showing CFLOBDDs performing better than \TIDDs in practice in the quantum-simulation domain.
    \sectref{related} discusses related work.
    \sectref{conclusion}
    presents some concluding remarks.


\section{Background and Overview}
\label{Se:BackgroundAndOverview}

\begin{figure}[tb!]
  \centering
  \begin{tabular}{c@{\hspace{8.0ex}}c}
    \includegraphics[align=c,scale=0.45]{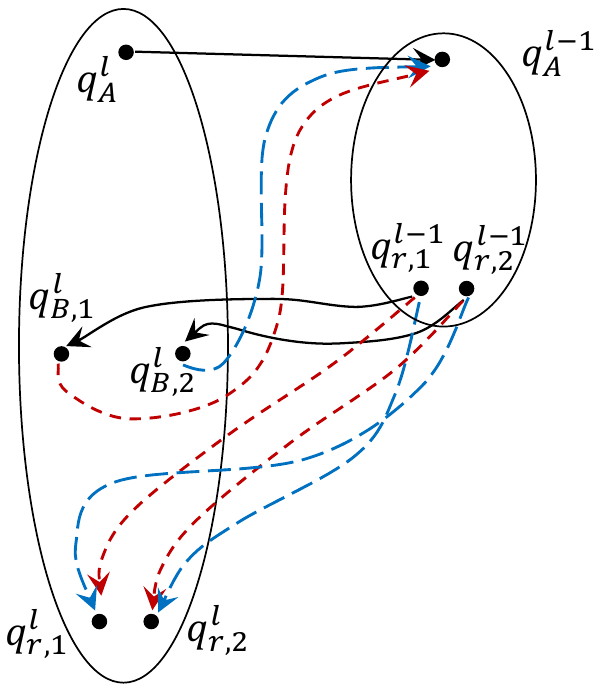}
    &
    \begin{tabular}{c}
      \includegraphics[scale=0.45]{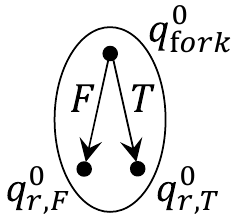}
      \\
      (b)
      \\
      \\
      \\
      \includegraphics[scale=0.45]{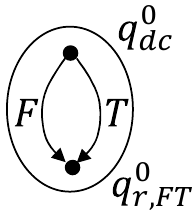}
      \\
      (c)
    \end{tabular}
    \\
    (a) &
    \\
    \\
    \multicolumn{2}{c}{
    \begin{tabular}{l|ll}
      \hline\hline
      \multirow{9}{*}{(a)} & \multirow{3}{*}{call transitions}     & $(q_A^l,\epsilon,q_A^{l-1}) \in \delta_c$ \\
                           &                                       & $(q_{B,1}^l,\epsilon,q_A^{l-1}) \in \delta_c$ \\
                           &                                       & $(q_{B,2}^l,\epsilon,q_A^{l-1}) \in \delta_c$ \\
                           \cline{2-3}
                           & \multirow{6}{*}{return transitions}   & $(q_{r,1}^{l-1},q_A^l,\epsilon,q_{B,1}^l) \in \delta_r$ \\
                           &                                       & $(q_{r,2}^{l-1},q_A^l,\epsilon,q_{B,2}^l) \in \delta_r$ \\
                           &                                       & $(q_{r,1}^{l-1},q_{B,1}^l,\epsilon,q_{r,1}^l) \in \delta_r$ \\
                           &                                       & $(q_{r,2}^{l-1},q_{B,1}^l,\epsilon,q_{r,2}^l) \in \delta_r$ \\
                           &                                       & $(q_{r,1}^{l-1},q_{B,2}^l,\epsilon,q_{r,2}^l) \in \delta_r$ \\
                           &                                       & $(q_{r,2}^{l-1},q_{B,2}^l,\epsilon,q_{r,1}^l) \in \delta_r$ \\
      \hline
      \multirow{2}{*}{(b)} & \multirow{2}{*}{internal transitions} & $(q^0_{\textrm{fork}}, F, q^0_{r,F}) \in \delta_i$ \\
                           &                      & $(q^0_{\textrm{fork}}, T, q^0_{r,T}) \in \delta_i$ \\
      \hline
      \multirow{2}{*}{(c)} & \multirow{2}{*}{internal transitions} & $(q^0_{\textrm{dc}}, F, q^0_{r,FT}) \in \delta_i$  \\
                           &                      & $(q^0_{\textrm{dc}}, T, q^0_{r,FT}) \in \delta_i$  \\
      \hline\hline
    \end{tabular}
    }
  \end{tabular}
  \caption{\protect \raggedright 
  (a) Encoding of a grouping's A-connection and B-connections as call transitions, and its return edges as return transitions of an NWA.
  The grouping is the one used to encode the family of Hadamard matrices $\HadamardFamily$.
  (b) and (c) Encoding of the two kinds of level-$0$ groupings as internal transitions of an NWA.}
  \label{Fi:ch6_CFLOBDDAsNWA}
\end{figure}

\subsection{Nested Words and Nested-Word Automata}
\label{Se:ch6_NWADefinition}

\begin{defn}[\cite{alur2009adding}]
\label{De:ch6_NestedWord}
A \textbf{nested word} $(w,\rightsquigarrow)$ over alphabet $\Sigma$
is an ordinary word $w \in \Sigma^*$, together with a \textbf{nesting relation}
$\rightsquigarrow$ of length $|w|$. $\rightsquigarrow$ is a collection of edges
(over the positions in $w$) that do not cross.
A nesting relation of length $l \geq 0$ is a subset of
$\{-\infty, 1, 2, \ldots, l \} \times \{1, 2, \ldots, l, +\infty\}$ such that
\begin{itemize}
  \item
    Nesting edges only go forwards: if $i \rightsquigarrow j$ then $i < j$.
  \item
    No two edges share a position: for $1 \leq i \leq l$,
    $|\{ j \mid i \rightsquigarrow j \}| \leq 1$
    and $|\{ j \mid j \rightsquigarrow i \}| \leq 1$.
  \item
    Edges do not cross: if $i \rightsquigarrow j$ and $i' \rightsquigarrow j'$,
    then one cannot have $i < i' \leq j < j'$.
\end{itemize}
When $i \rightsquigarrow j$ holds, for $1 \leq i \leq l$,
$i$ is called a \textbf{call} position; if $i \rightsquigarrow +\infty$,
then $i$ is a \textbf{pending call}; otherwise $i$ is a \textbf{matched call},
and the unique position $j$ such that $i \rightsquigarrow j$ is called
its \textbf{return successor}.
Similarly, when $i \rightsquigarrow j$ holds, for $1 \leq j \leq l$,
$j$ is a \textbf{return} position; if $-\infty \rightsquigarrow j$,
then $j$ is a \textbf{pending return}, otherwise $j$ is a \textbf{matched return},
and the unique position $i$ such that $i \rightsquigarrow j$ is called
its \textbf{call predecessor}.
A position $1 \leq i \leq l$ that is neither a call nor a return
is an \textbf{internal} position.

$\textbf{MatchedNW}$ denotes the set of nested words that have no
pending calls or returns.
$\textbf{NWPrefix}$ denotes the set of nested words that have no
pending returns.

A \textbf{nested word automaton} (NWA) $A$ is a 5-tuple $(Q, \Sigma,
q_0, \delta, F)$,
where $Q$ is a finite set of states, $\Sigma$ is a
finite alphabet, $q_0 \in Q$ is the initial state, $F \subseteq Q$ is a set
of final states, and $\delta$ is a transition relation. The transition
relation $\delta$ consists of three components, $(\delta_c, \delta_i,
\delta_r)$, where
\begin{itemize}
  \item
    $\delta_i \subseteq Q \times \Sigma \times Q$ is the transition
    relation for internal positions.
  \item
    $\delta_c \subseteq Q \times \Sigma \times Q$ is the transition
    relation for call positions.
  \item
    $\delta_r \subseteq Q \times Q \times \Sigma \times Q$ is the
    transition relation for return positions.
\end{itemize}

Starting from $q_0$, an NWA $A$ reads a nested word
$\textit{nw} = (w,\rightsquigarrow)$ from left to
right, and performs transitions (possibly non-deterministically)
according to the input symbol and $\rightsquigarrow$.
If $A$ is in state $q$ when reading input symbol $\sigma$ at position $i$ in $w$,
and $i$ is an internal or call position, $A$ makes a transition to $q'$
using $(q,\sigma,q') \in \delta_i$ or $(q,\sigma,q') \in \delta_c$, respectively.
If $i$ is a return position, let $k$ be the call predecessor of $i$,
and $q_c$ be the state $A$ was in just before the transition it made on the
$k^{\text{th}}$ symbol;
$A$ uses $(q,q_c,\sigma,q') \in \delta_r$ to make a transition to $q'$.
If, after reading $\textit{nw}$, $A$ is in a state $q \in F$,
then $A$ \textbf{accepts} $\textit{nw}$.
\end{defn}

The automaton $A$ processes the nested word $nw$ according to a linear order of the characters of $nw$, along with obeying the hierarchy in $nw$.
As discussed in~\cite{alur2009adding}, the state propagated along the linear edges (internal transitions) is the same as in the case
of a standard word automaton. At a call position, the nesting edge is considered to determine the hierarchical state,
and at a return position, the automaton determines the new state based on the linear (internal) and return edges.

\figref{ch6_CFLOBDDAsNWA} illustrates a schema by which a CFLOBDD can be translated to an NWA $M$.
Each matched path through the CFLOBDD corresponds to a nested word in MatchedNW for $M$. 
The matched-path principle is obeyed because of the ability of an NWA to ``peek'' at the state of the most-recent ``call'' to match a return edge with the appropriate preceding A-connection or B-connection.
All transitions taken at a level $\geq 1$ are $\epsilon$-transitions (\figref{ch6_CFLOBDDAsNWA}a).
The only transitions that consume an alphabet symbol are the $F$ and $T$ transitions of the level-$0$ fork grouping (\figref{ch6_CFLOBDDAsNWA}b) and the $F$ and $T$ transitions of the level-$0$ don't-care grouping (\figref{ch6_CFLOBDDAsNWA}c).

\begin{figure}[t!]
    \centering
    \begin{tabular}{c}
     \includegraphics[scale=0.4]{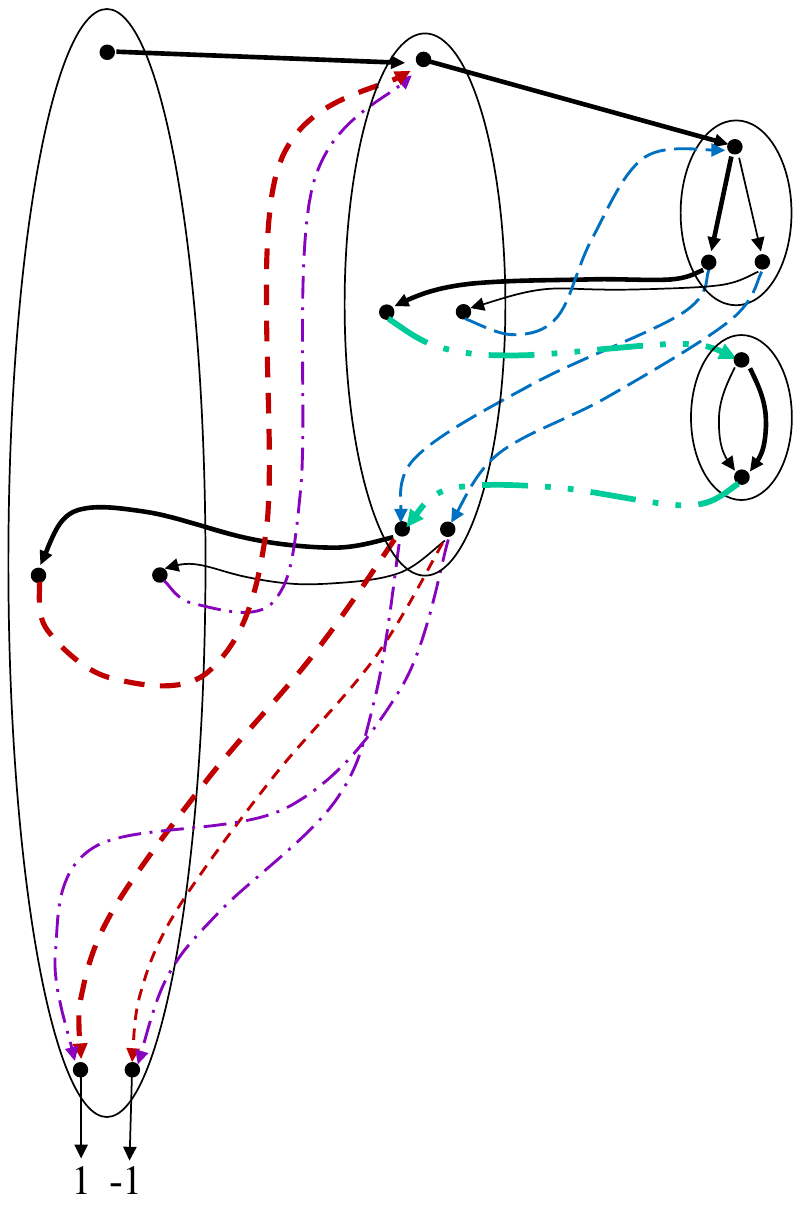} \\
    (a) \\
    \\
    \includegraphics[scale=0.46]{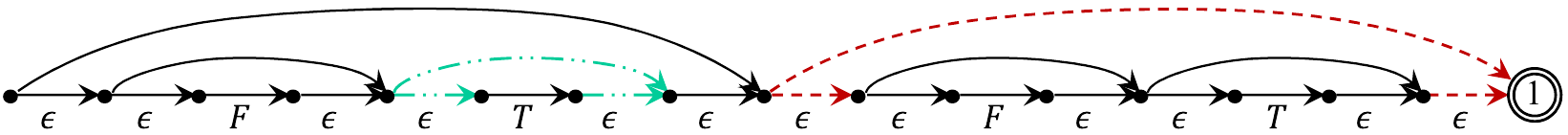} \\
    (b)
    \end{tabular}
    \caption{\protect \raggedright 
    (a) The CFLOBDD for the Hadamard matrix $H_4$, with the variable ordering $\langle x_0, y_0, x_1, y_1 \rangle$.
    The matched path for $[x_0 \mapsto F, y_0 \mapsto T, x_1 \mapsto F, y_1 \mapsto T]$, which corresponds to $H_4[0,3]$ (with value $1$), is shown in bold.
    (b) The nested word for the path for  $[x_0 \mapsto F, y_0 \mapsto T, x_1 \mapsto F, y_1 \mapsto T]$.
    }
    \label{Fi:ch3_NestedWordExample}
\end{figure}

\figref{ch3_NestedWordExample} shows the nested word that corresponds to the path for the assignment  $[x_0 \mapsto F, y_0 \mapsto T, x_1 \mapsto F, y_1 \mapsto T]$ in the CFLOBDD for the Hadamard matrix $H_4 = \begin{bmatrix}
    1 & 1 & 1 & 1\\
    1 & -1 & 1 & -1\\
    1 & 1 & -1 & -1\\
    1 & -1 & -1 & 1\\
\end{bmatrix}$, with the variable ordering $\langle x_0, y_0, x_1, y_1 \rangle$, where $x_0, x_1$ correspond to the row variables and $y_0, y_1$ correspond to the column variables.

For each assignment, the state after a given prefix $p$ is a function of the bindings in $p$ seen so far,
making CFLOBDDs depend on 
the linear order of the assignment.
But it is worth pointing out that the hierarchical structure interacts with the linear order to allow proto-CFLOBDDs to be shared.
For instance, consider the three level-2-to-level-1 calls in \figref{ch3_NestedWordExample} (where the three calling contexts are the level-2 grouping's entry vertex, its first middle vertex, and its second middle vertex).
It allows the level-1 proto-CFLOBDD to be shared at the three calls---with all three transitioning to the state represented by the entry vertex of the level-1 grouping.
At the exit vertices of the level-1 grouping, the various return transitions transfer control to the appropriate states in the level-2 grouping according to the calling context.

\subsection{A Hierarchical Decision Diagram that Lacks Linear Structure}
\label{Se:AHierarchicalDecisionDiagramThatLacksLinearStructure}

To understand whether the functions represented by CFLOBDDs exploit the 
linear structure of words (i.e., assignments),
we introduce a new kind of decision diagram that is similar to CFLOBDDs, but without their NWA-like
linear structure.
This approach results in the new decision diagrams possessing only the hierarchical properties, and hence
they
can be viewed as a restricted form of deterministic tree automata.
More formally, we introduce Tree-automata Inspired Decision Diagrams
(\TIDDs)---acyclic deterministic tree-automata to represent Boolean functions---which
run on trees whose leaves are values of the different variables of the given function.
That is, a tree represents an assignment, and for a pseudo-Boolean function $f$, the value that the \TIDD for $f$ gives to the tree for assignment $a$ is $f(a)$.
We want \TIDDs to be as close to CFLOBDDs as possible without the linearity in CFLOBDDs;
hence, \TIDDs also follow 
an analogue of
the ``Contextual-Interpretation'' principle obeyed by CFLOBDDs:
if two sub-functions defined over different groups of variables are equal, then they are represented by the same substructure in 
the data structure that represents
a \TIDD.
(This point should become clearer in \sectref{ch6_tidds}.)

In particular, we want to understand how deterministic tree automata can represent Boolean functions, and how they compare with CFLOBDDs in terms of the representation size and
cost of performing
operations.

For the representation of an assignment $a$, we use a perfect binary tree $t_a$ such that the linear sequence (left-to-right) formed by the leaves of $t_a$ is equal to the assignment $a$.
The tree structure is equivalent to the grammar followed by CFLOBDDs, so
the ideas behind \TIDDs could be generalized to allow the user to define the tree structure to use (similar to the generalization of CFLOBDDs discussed in the Future Work of~\cite{2211.06818}).

\begin{figure}[h!]
    \centering
    \begin{subfigure}{0.45\linewidth}
    \includegraphics[width=\linewidth]{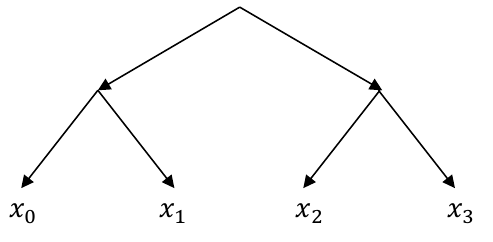}
    \caption{Tree with four variables}
    \end{subfigure}
    \begin{subfigure}{0.45\linewidth}
    \includegraphics[width=\linewidth]{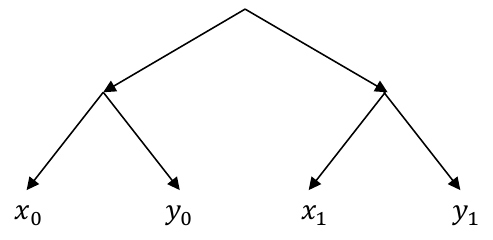}
    \caption{Interleaved tree with four variables}
    \end{subfigure}
    \caption{Tree representation of the variables of a function}
    \label{Fi:ch6_sample-tree}
\end{figure}

However, because CFLOBDDs follow a balanced grammar, \TIDDs accept only perfect binary trees, where the tree's leaves contain the function's variables and the internal nodes do not contain any symbols/variables.
\figref{ch6_sample-tree}(a) shows the tree for a function over 4 variables -- $\langle x_0, x_1, x_2, x_3 \rangle$, and \figref{ch6_sample-tree}(b) shows the tree for a function (e.g., a matrix) over two sets of interleaved variables -- $\langle x_0, y_0, x_1, y_1 \rangle$.

We show that CFLOBDDs indeed make use of the linear structure and can represent functions
more
efficiently than \TIDDs.

\section{Tree-Automata Inspired Decision Diagrams (\TIDDs)}
\label{Se:ch6_tidds}

\subsection{Basic Structure}
\label{Se:ch6_BasicStructure}

Because \TIDDs are a special version of deterministic finite tree automata that represent Boolean functions, we will define \TIDDs in the same way as one defines a tree automaton, with small differences.
Generally, a tree automaton is a tuple $M = (\mathcal{Q}, \mathcal{Q}_f, \mathcal{F}, \Delta)$.
$M$ accepts a language of trees $\mathcal{L}$, where $\mathcal{Q}$ is the set of states, $\mathcal{Q}_f$ is the set of final states ($\mathcal{Q}_f \subset \mathcal{Q}$), $\mathcal{F}$ is the set of symbols and $\Delta$ represents the transition relation. 
Because we want to represent functions, there is no notion of ``acceptance.'' All the states are accepted, but every final state has an associated ``value.'' Note that we are only dealing with acyclic deterministic tree automata.

We want to represent functions using \TIDDs. An assignment $a$ of the variables of the function is converted to a tree $t_a$ that the tree automaton runs on.
If a function $f$ is defined over $n$ Boolean variables, then the corresponding tree $t_a$ is a 
perfect binary tree of
height $\log n$. The leaves of $t_a$ are the Boolean variables, and the internal nodes have one symbol $\Omega$.

\begin{defn}\label{De:ch6_TIDDDefn}
    We define a \TIDD as a deterministic, acyclic, bottom-up, tree automaton, divided into ``levels'' with a tuple $M = (l, \mathcal{Q} = \mathcal{Q}_0 \cup \mathcal{Q}_1 \cup \ldots \cup \mathcal{Q}_l, \mathcal{Q}_f, \mathcal{F}, \Delta, \mathcal{V})$ and the following constraints:
    \begin{itemize}
        \item The number of levels of states in a \TIDD is $l+1$.
        \item $\mathcal{Q}$ is the set of states. Each $Q_i$, for $i \in \{ 0,\ldots, l\}$, is the set of states at level $i$, called a state-layer.
        \item $\mathcal{Q}_f$ is the set of final states with $\mathcal{Q}_f = \mathcal{Q}_l$.
        \item $\mathcal{F}$ is the alphabet, i.e., set of symbols. $\mathcal{F} = \{0,1, \Omega\}$. $\textbf{0}$, $\textbf{1}$ correspond to the symbols at the leaves of the tree, and an implicit symbol -- $\Omega$, with arity 2 -- corresponds to an internal symbol of the tree.
        \item $\Delta$ is the transition relation. Because the automaton is deterministic, the transition relation $\Delta$ becomes a transition function $\delta$ of the form: $\delta: {\mathcal{Q}_i} \times {\mathcal{Q}_i} \rightarrow {\mathcal{Q}_{i+1}}$, for $i \in \{0,\ldots, l-1\}$. 
        At level-$0$, $\delta: \mathcal{F} \rightarrow Q_0$. We assign a fixed ``rejected'' state for $\delta$ on symbol $\Omega$ at
        level-$0$---i.e., no leaf is labeled by $\Omega$,
        and do not explicitly discuss it when
        defining
        a \TIDD.
        We refer to 
        an input-output triple of $\delta$ (or input-output pair, in the case of level-0) as either a ``transition'' or a ``hyper-edge'' (usually shortened to ``edge'').
        \item $\mathcal{V}$ is the value function from every final state ($\mathcal{Q}_f$) to a value $v \in \mathcal{D}$ (some domain), i.e., $\mathcal{V}: \mathcal{Q}_f \rightarrow \mathcal{D}$. Every assignment $a$ of variables of a function $f$ has a corresponding tree $t_a$ with the variables as the leaves of the tree. Hence, the value of the final state reached by the running $M$ on $t_a$ is equal to the value of the function $f$ at $a$.
    \end{itemize}

\TIDDs must also satisfy the following properties.
\begin{enumerate}
  \item
    Each assignment has a value (existence of a run): for every assignment-tree $t$, $M$ does not get stuck: there is a run of $M$ that leads to some final state in $\mathcal{Q}_f$.
  \item
    Uniqueness of the representation of $\delta$:
    \begin{enumerate}[label=(\roman*)]
      \item
        At each level $i$, there is a total order on the states in $\mathcal{Q}_i$.
        Let $\textit{Seq}_{\mathcal{Q}_i}$ denote the vector that list the states in that order.
        The names of states at level $i$ are chosen according to the total order:
        $[q^i_1, \ldots, q^i_{|\textit{Seq}_{\mathcal{Q}_i}|}]$.
      \item
        \label{It:LevelZeroStates}
        At level 0, the state assigned when the symbol is a 0 is $q^0_0$; the state assigned when the symbol is a 1 is $q^0_1$.
        $\textit{Seq}_{\mathcal{Q}_0} =_\textit{df} [q^0_0, q^0_1]$.
      \item
      \label{It:LevelIPlus1States}
        The total order on states at level $i+1$ equals the order in which they are mentioned in the sequence
        \[
          \textit{Seq}_{\mathcal{Q}_{i+1}} = [ \delta(q_l, q_r) \mid (q_l, q_r) \in \textit{Seq}_{\mathcal{Q}_i} \otimes \textit{Seq}_{\mathcal{Q}_i}]
        \]
        (where $\textit{Seq}_{\mathcal{Q}_i} \otimes \textit{Seq}_{\mathcal{Q}_i}$ is a Kronecker-product operation that pairs sequence elements, and creates a total order on the pairs of states in $\mathcal{Q}_i$).
      \item
        $\mathcal{V}$ is one-to-one and onto $\mathcal{D}$.
      \item 
        \label{It:Congruence}
        At each level $i$, any two different states can be distinguished by how they interact with at least one other level-$i$ state via the transition function $\delta$:
        \[
          \forall q_j, q_k \in \mathcal{Q}_i . \left( q_j \neq q_k
          \implies
          \exists q \in \mathcal{Q}_i .
            \delta(q_j, q) \neq \delta(q_k, q)
            \lor
            \delta(q, q_j) \neq \delta(q, q_k) \right).
        \]
    \end{enumerate}
\end{enumerate}
\end{defn}

\begin{exmp}\label{Exa:ch6_hadamard}

\begin{figure}[tb!]
    \centering
    \begin{subfigure}[t]{0.2\linewidth}
        \includegraphics[width=.75\linewidth]{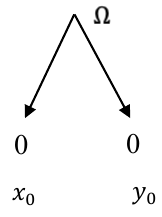}
    \end{subfigure}
    \hspace{2ex}
    \begin{subfigure}[t]{0.2\linewidth}
        \includegraphics[width=.75\linewidth]{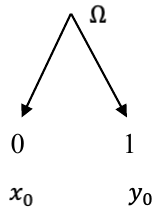}
    \end{subfigure}
    \hspace{2ex}
    \begin{subfigure}[t]{0.2\linewidth}
        \includegraphics[width=.75\linewidth]{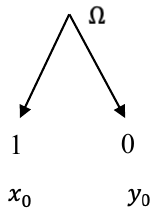}
    \end{subfigure}
    \hspace{2ex}
    \begin{subfigure}[t]{0.2\linewidth}
        \includegraphics[width=.75\linewidth]{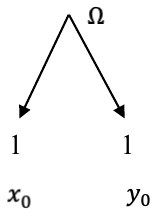}
    \end{subfigure}
    \caption{The set of trees $\mathcal{T_H}$ representing the assignments of $H_2$ with two variables $x_0, y_0$ as the leaves of the tree.}
    \label{Fi:ch6_walsh1_tidd_tree}
\end{figure}

\begin{figure}[tb!]
    \centering
    \begin{subfigure}[t]{0.43\linewidth}
        \includegraphics[width=0.65\linewidth]{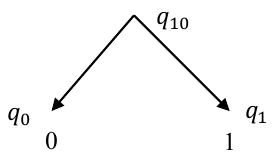}
        \caption{
        The diagram
        shows states $q_0$, $q_1$, and $q_{10}$, and transitions [$0 \rightarrow q_0$, $1 \rightarrow q_1$, $(q_0, q_1) \rightarrow q_{10}$] of $M_{H_2}$ for the tree representing the assignment $\langle x_0 \mapsto 0, y_0 \mapsto 1 \rangle$.}
        \label{Fi:ch6_walsh1_tidd}
    \end{subfigure}
    \hspace{2.0ex}
    \begin{subfigure}[t]{0.5\linewidth}
        \includegraphics[width=0.85\linewidth]{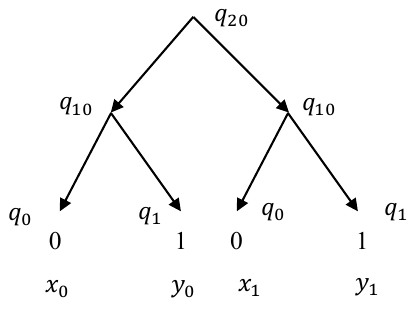}
        \caption{
        The diagram
        shows states $q_0$, $q_1$, $q_{10}$, and $q_{20}$, and transitions [$0 \rightarrow q_0$, $1 \rightarrow q_1$, $(q_0, q_1) \rightarrow q_{10}$, $(q_{10}, q_{10}) \rightarrow q_{20}$] of $M_{H_4}$ for the tree representing the assignment $\langle x_0 \mapsto 0, y_0 \mapsto 1, x_1 \mapsto 0, y_1 \mapsto 1 \rangle$.}
        \label{Fi:ch6_walsh2_tidd}
    \end{subfigure}
    \caption{
    The diagrams
    show the runs of tree automata (\TIDDs) $M_{H_2}$ and $M_{H_4}$ on two trees corresponding to two assignments 
    for the functions (matrices)
    $H_2$ and $H_4$.}
    \label{Fi:ch6_walsh_tidd}
\end{figure}

Let us consider a family of Hadamard matrices, $\mathcal{H} = \{H_{2^i} \mid i \geq 1 \}$,
and for $i \geq 1$, $H_{2^{i+1}} = H_{2^i} \tensor H_{2^i}$, where $H_2 = \begin{bNiceArray}{cc}[first-col,first-row]
                       & 0 & 1\\
                     0 & 1 & 1\\
                     1 & 1 & -1\\
                \end{bNiceArray}$.

$H_2$ is a function over two Boolean variables -- $x_0$ for 
row-index variables
and $y_0$ for 
column-index variables.
Hence, $H_2$ has four assignments corresponding to $\langle x_0, y_0 \rangle = \{00,01,10,11\}$.
The set of trees $\mathcal{T_H}$
that represent these assignments are the four binary trees with two leaves
shown in~\figref{ch6_walsh1_tidd_tree}.
The \TIDD representation of $H_2$ would be a tuple $M_{H_2} = (l = 1, \mathcal{Q} = {\mathcal{Q}_0} \cup {\mathcal{Q}_1}, \mathcal{Q}_f = {\mathcal{Q}_1}, \mathcal{F} = \{0,1,\Omega\}, \delta, \mathcal{V})$, with
\begin{itemize}
    \item $\mathcal{Q}_0 = \{q_0, q_1\}$, $\mathcal{Q}_1 = \{q_{10}, q_{11}\}$
    \item $\delta$ function:
    \[
    \begin{array}{cccc}
         \multicolumn{2}{c}{\delta(\mathbf{0}) \rightarrow q_0} & \multicolumn{2}{c}{\delta(\mathbf{1}) \rightarrow q_1}\\
         \delta(q_{0}, q_{0}) = q_{10} & \delta(q_{0}, q_{1}) = q_{10} & \delta(q_{1}, q_{0}) = q_{10} & \delta(q_{1}, q_{1}) = q_{11}\\
    \end{array}
    \]
    \item $\mathcal{V}(q_{10}) = 1$, $\mathcal{V}(q_{11}) = -1$
\end{itemize}

We can observe that $M_{H_2}$ on assignments $\{00,01,10\}$ leads to the value $1 = H_2[0,0] = H_2[0,1] = H_2[1,0]$, and assignment $\{11\}$ leads to the value $-1 = H_2[1,1]$.
The \TIDD representation of $H_2$ has 4 states and 6 edges.
\figref{ch6_walsh1_tidd} shows the run of $M_{H_2}$ on the tree with leaves $\langle 0, 1 \rangle$, which corresponds to the assignment $[x_0 \mapsto 0, y_0 \mapsto 1]$. We can observe that the states on reading the symbols at the leaves 
are
$q_0$ for $0$ and $q_1$ for $1$,
and similarly, the state transitioned to at the next level
is
$(q_0, q_1) \rightarrow q_{10}$.

Let us consider the next matrix in this series $H_4 = H_2 \tensor H_2 = \begin{bNiceArray}{cc}[first-col,first-row]
                       & 0 & 1\\
                     0 & H_2 & H_2\\
                     1 & H_2 & -H_2\\
                \end{bNiceArray} = \begin{bNiceArray}{cccc}[first-col,first-row]
                       & 00 & 01 & 10 & 11\\
                     00 & 1 & 1 & 1 & 1\\
                     01 & 1 & -1 & 1 & -1\\
                     10 & 1 & 1 & -1 & -1\\
                     11 & 1 & -1 & -1 & 1\\
                \end{bNiceArray}$.
The \TIDD representation of $H_4$ is a tuple $M_{H_4} = (l = 2, \mathcal{Q} = {\mathcal{Q}_0} \cup {\mathcal{Q}_1} \cup {\mathcal{Q}_2}, \mathcal{Q}_f = {\mathcal{Q}_2}, \mathcal{F} = \{0,1,\Omega\}, \delta, \mathcal{V})$, with
\begin{itemize}
    \item $\mathcal{Q}_0 = \{q_0, q_1\}$, $\mathcal{Q}_1 = \{q_{10}, q_{11}\}$, and $\mathcal{Q}_2 = \{q_{20}, q_{21}\}$
    \item $\delta$ function:
    \[
    \begin{array}{cccc}
         \multicolumn{2}{c}{\delta(\mathbf{0}) \rightarrow q_0} & \multicolumn{2}{c}{\delta(\mathbf{1}) \rightarrow q_1}\\
         \delta(q_{0}, q_{0}) = q_{10} & \delta(q_{0}, q_{1}) = q_{10} & \delta(q_{1}, q_{0}) = q_{10} & \delta(q_{1}, q_{1}) = q_{11}\\
         \delta(q_{10}, q_{10}) = q_{20} & \delta(q_{10}, q_{11}) = q_{21} & \delta(q_{11}, q_{10}) = q_{21} & \delta(q_{11}, q_{11}) = q_{20}\\
    \end{array}
    \]
    \item $\mathcal{V}(q_{20}) = 1$, $\mathcal{V}(q_{21}) = -1$
\end{itemize}

The \TIDD representation of $H_2$ has 6 states and 10 edges.
We can observe that ${\mathcal{Q}_0}, {\mathcal{Q}_1}$ of $M_{H_4}$ and $M_{H_2}$ are the same. 
Similarly, the transition function between states of $\mathcal{Q}_0$ and $\mathcal{Q}_1$ of $M_{H_4}$ and $M_{H_2}$
are also the same.
\figref{ch6_walsh2_tidd} shows the run of $M_{H_4}$ on the tree with leaves $\langle 0, 1, 0, 1 \rangle$ that corresponds to the assignment $[x_0 \mapsto 0, y_0 \mapsto 1, x_1 \mapsto 0, y_1 \mapsto 1]$. We can observe that the states on reading the symbols at the leaves would be $q_0$ for $0$ and $q_1$ for $1$, similarly, the states transitioned to at the next level would be $(q_0, q_1) \rightarrow q_{10}$,
and $(q_{10}, q_{11}) \rightarrow q_{20}$.

On the same lines, the \TIDD representation of $H_{2^i}$ is a tuple $M_{H_{2^i}} = (l = i, \mathcal{Q} = {\mathcal{Q}_0} \cup {\mathcal{Q}_1} \cup \ldots \cup {\mathcal{Q}_i}, \mathcal{Q}_f = {\mathcal{Q}_i}, \mathcal{F} = \{0,1,\Omega\}, \delta, \mathcal{V})$, with
\begin{itemize}
    \item $\mathcal{Q}_0 = \{q_0, q_1\}$, $\mathcal{Q}_j = \{q_{j0}, q_{j1}\}$, for $j = 1\ldots i$.
    \item $\delta$ function, for $j = 1\ldots i-1$:
    \[
    \begin{array}{cccc}
         \multicolumn{2}{c}{\delta(\mathbf{0}) \rightarrow q_0} & \multicolumn{2}{c}{\delta(\mathbf{1}) \rightarrow q_1}\\
         \delta(q_{0}, q_{0}) = q_{10} & \delta(q_{0}, q_{1}) = q_{10} & \delta(q_{1}, q_{0}) = q_{10} & \delta(q_{1}, q_{1}) = q_{11}\\
         \delta(q_{10}, q_{10}) = q_{20} & \delta(q_{10}, q_{11}) = q_{21} & \delta(q_{11}, q_{10}) = q_{21} & \delta(q_{11}, q_{11}) = q_{20}\\
         \vdots\\
         \delta(q_{j0}, q_{j0}) = q_{(j+1)0} & \delta(q_{j0}, q_{j1}) = q_{(j+1)1} & \delta(q_{j1}, q_{j0}) = q_{(j+1)1} & \delta(q_{j1}, q_{j1}) = q_{(j+1)0}\\
    \end{array}
    \]
    \item $\mathcal{V}(q_{i0}) = 1$, $\mathcal{V}(q_{i1}) = -1$
\end{itemize}
The \TIDD representation of $H_{2^i}$ has $2i + 2$ states and $4i + 2$ edges.
\end{exmp}

\paragraph{Discussion.}
Note that the trees in~\figref{ch6_walsh_tidd} do not show the symbol $\Omega$ at the internal nodes; we assume it to be implicit.
Also, we can observe that, unlike general tree automata, the automata do not have
level-0
transitions specific to variables
(or, more precisely, to the \emph{position} of a variable in an assignment), but operate solely on the
symbols $0$ and $1$.
The reason behind this technique is to make the \TIDD representation as similar to CFLOBDDs as possible.
If a \TIDD 
were to use
a different state for every variable, then the \TIDD representation of $H_{2^i}$ would have size $\mathcal{O}(2^i)$, whereas
the 
size of the representation of $H_{2^i}$ used in our work is $\mathcal{O}(i)$.

\subsection{Canoncity}
\label{Se:ch6_Canonicity}
As shown in \cite{comon2008tree}, the Myhill–Nerode theorem for tree automata implies the existence and uniqueness (up to isomorphism) of a minimal deterministic tree automaton recognizing a given tree language.
We will use the construction following the Myhill-Nerode theorem to show that \TIDDs are minimal and thereby canonical.

The Myhill--Nerode equivalence for trees induces a congruence on the set of trees.
Two trees are equivalent if and only if they are indistinguishable by all contexts:
that is, for any context $C[\;]$, plugging either tree into $C[\;]$ yields trees
that are either both accepted or both rejected.
Equivalently, a context of a subtree $t'$ in a larger tree $t$ is the surrounding
tree with a distinguished hole, and $t'$ may be replaced by another tree in that hole.
The minimal deterministic bottom-up tree automaton recognizing a tree
language is obtained by taking the Myhill--Nerode equivalence classes as states.
Transitions are defined as follows:
for a symbol $f$ of arity $k$ and equivalence classes
$[t_1], \ldots, [t_k]$, the transition on $f$ maps
$([t_1], \ldots, [t_k])$ to the equivalence class of the tree
$f(t_1, \ldots, t_k)$.

Formally, two states are equivalent if the trees represented by their
equivalence classes are indistinguishable under all contexts.
Equivalently, in terms of the transition function, if two states $q_1$ and $q_2$ are equivalent
--- and hence can be merged --- then their transition behavior must coincide.
That is,
\[
\forall q,\quad
[\delta(q_1, q)] \equiv [\delta(q_2, q)]
\;\wedge\;
[\delta(q, q_1)] \equiv [\delta(q, q_2)],
\]
where $[\cdot]$ denotes the equivalence class of a state.

We observe that this condition on the transition function
has been incorporated into constraint 2\ref{It:Congruence} that is
imposed on the transition structure of
\TIDDs (page~\pageref{It:Congruence}).
As a result, if the
states at level $i+1$ are minimal, then the states at level $i$ are also minimal,
because no two distinct states at level $i$ can satisfy the above equivalence;
that is, there are no pairs of states to merge in a \TIDD.

At the top-most level $l$, the number of states is determined by the value
function $\mathcal{V}$, which is a bijection onto the domain $\mathcal{D}$.
Consequently, the number of states at level $l$ is minimal.

Putting these observations together, we have shown that the states at the
top-most level are minimal, and that the transition function of \TIDD---being
equivalent to that of a minimal deterministic bottom-up tree automaton---ensures that minimality propagates from level $i+1$ to level $i$. Therefore,
\TIDDs are minimal by construction.

We now state the canonicity theorem for \TIDDs.

\begin{thm}[Canonicity of \TIDDs]\label{The:ch6_Canonicity}
If $M_1$ and $M_2$ are level-$l$ \TIDDs for the same Boolean function over $n=2^l$ Boolean variables, and $M_1$ and $M_2$ use the same variable ordering, then $M_1$ and $M_2$ are isomorphic.
\end{thm}

To prove this theorem, we first introduce and discuss two definitions: 
(i) the \emph{down language} of a state, and (ii) the \emph{down-assignment-language} of a state.

\paragraph{Down Language.}
Let $T$ represent the set of all possible perfect binary trees with $\mathbf{0}, \mathbf{1}$ as the leaves of the trees.
For an automaton $A$ that runs on $T$, the language accepted by $A$ is 
\[
L(A) = \{t \in T \mid m_A(t) \in Q_f\}
\]
where, $m_A(t)$ is the state obtained by running $A$ on $t$.
In the same way, \cite{guellouma2016efficient}
define the
down-language of a state $q$ of the automaton $A$:
\[
L^{\downarrow}(q) = \{t \in T \mid m_A(t) = q \}
\]
$L^{\downarrow}(q)$
consists of
the set of trees on which automaton $A$ reaches state $q$.\footnote{
  For a \TIDD, $L(A) = T$ because, when run on the tree for an arbitrary assignment, the \TIDD reaches \emph{some} final state (\defref{ch6_TIDDDefn}).
  In contrast, $L^{\downarrow}(q)$ allows us to make distinctions among different assignment-trees based on the
  state $q$
  reached.
}

We will now define a slight variant of the down-language. We will define a new term called
``down-assignment-language'' $L^{\downarrow_f}$. For a Boolean function $f$ over $n$ variables,
 there exists $2^n$ perfect binary trees.
 Every perfect binary tree $t$ can be considered as a one-step-deeper tree over two perfect binary trees $t_a$ and $t_b$ such that $t$ is formed using $t_a$ as the left-child and $t_b$ as the right-child, i.e., $t = t_a || t_b$.
 Let $T$ represent all such unique binary (sub-)trees of the $2^n$ perfect binary trees
 corresponding to the function $f$. Every tree $t \in T$ has a corresponding word $w$ obtained by sequentially concatenating the leaves of $t$; let us denote this operation by $w = TW(t)$ and $t = TW^{-1}(w)$.
 Let $T_w$ represent all such unique words obtained from the trees in $T$.
 The ``down-assignment-language'' $L^{\downarrow_f}$ of a state $q$ of automaton $A$ is defined as
 \[
 L^{\downarrow_f}(q) { =_{\textit{df}} } \{ w \in T_w \mid m_A(TW^{-1}(w)) = q\}.
 \]

This 
definition
implies that $L^{\downarrow_f}$ of a state $q$ represents the set of words whose corresponding trees lead to state $q$ on running automaton $A$.
These words are, in fact, partial assignments of the function $f$. So in other words,
if state $q$ belongs to $\mathcal{Q}_i$, where $i$ is the height of the state layer,
$L^{\downarrow_f}(q)$ represents the set of Boolean strings of length $2^i$ that lead to state $q$ on running automaton $A$.

For a Boolean function $f$, because \TIDDs are deterministic, every $w \in T_w$ and correspondingly, every $t \in T$ belong to the $L^{\downarrow_f}$ of exactly one state, i.e.,
let $|.|$ represent the length of a word, and let $P(i)$ represent the set of all Boolean-words of length $2^i$, then
\[
\forall q \in \mathcal{Q}_i, \forall w \in L^{\downarrow_f}(q), |w| = 2^i
\]
\[
\forall w \in P(i), \exists!\, q \in \mathcal{Q}_i \text{ such that } w \in L^{\downarrow_f}(q)
\]

The states ($p$) of a \TIDD at every level-$i$ partition the space of $P(i)$ into $p$ partitions. That is, they satisfy the following properties:
\begin{enumerate}[label=(\roman*)]
    \item Any two states at the same level have different down-assignment-languages. Let the \TIDD have $l+1$ levels,
        \[
            \forall i \in \{0\ldots l \}, \forall q_j,q_k \in \mathcal{Q}_i, q_j \neq q_k \implies L^{\downarrow_f}(q_j) \cap L^{\downarrow_f}(q_k) = \phi
        \]
    \item The union of down-assignment-languages of all states at level-$i$ is equal to $P(i)$.
        \[
            \forall i \in \{0\ldots l\}, \bigcup\limits_{q \in \mathcal{Q}_i}L^{\downarrow_f}(q) = P(i)
        \]
\end{enumerate}

We will now give the proof for~\theoref{ch6_Canonicity}.

\begin{proof}

Consider two \TIDDs $M_1$ and $M_2$ that represent the same pseudo-Boolean function.

\paragraph{Base case: Level 0.}
By condition 2\ref{It:LevelZeroStates} that is imposed on the transition structure of \TIDDs (page~\pageref{It:LevelZeroStates}), $M_1$ and $M_2$ have the same level-0 states, namely, $q_0^0$ when a symbol is 0, and $q_1^0$ when a symbol is 1, with $\textit{Seq}_{\mathcal{Q}_0} = [q^0_0, q^0_1]$.

For use in satisfying the hypothesis of the inductive step, note that the corresponding down-assignment-languages are identical:
\[
  L^{\downarrow_{M_1}}(\textit{Seq}_{\mathcal{Q}_0}[j])
  =
  L^{\downarrow_{M_2}}(\textit{Seq}_{\mathcal{Q}_0}[j]), \textrm{ for } j \in \{ 0,1 \}.
\]

\paragraph{Inductive step: Level $i$ to level $i+1$.}
Inductive Hypothesis:
Assume that at level $i$ the sequences of states in $M_1$ and $M_2$ are such that (i) they are the same length;
(ii) states in the two \TIDDs are named according to position: $\textit{Seq}_{\mathcal{Q}_i} = [q_0^i, \ldots, q_{|\textit{Seq}_{\mathcal{Q}_i}|-1}^i]$; and
(iii) for all $0 \leq j \leq |\textit{Seq}_{\mathcal{Q}_i}|-1$, the corresponding down-assignment-languages are identical:
\[
  L^{\downarrow_{M_1}}(\textit{Seq}_{\mathcal{Q}_i}[j])
  =
  L^{\downarrow_{M_2}}(\textit{Seq}_{\mathcal{Q}_i}[j]).
\]

Consider the states at level~$i+1$.
The states of $M_1$ and $M_2$ at level~$i+1$ are uniquely determined by the transition
functions from level~$i$ to level~$i+1$.
Because $M_1$ and $M_2$ are minimal, their transition functions coincide.
Consequently, the number of states at level~$i+1$ in $M_1$ and $M_2$ is the same; in particular,

\begin{enumerate}
    \item[(i)] the sequences of states in $M_1$ and $M_2$ at level~$i+1$ have the same length.
\end{enumerate}

The ordering of states at level~$i+1$ is determined by
condition~2\ref{It:LevelIPlus1States}
on the transition structure of \TIDDs
(page~\pageref{It:LevelIPlus1States}).
Specifically, the sequence of states at level~$i+1$ is given by
\[
\textit{Seq}_{\mathcal{Q}_{i+1}}
=
\bigl[\, \delta(q_l, q_r) \;\big|\;
(q_l, q_r) \in \textit{Seq}_{\mathcal{Q}_i} \otimes \textit{Seq}_{\mathcal{Q}_i}
\,\bigr],
\]
where $\textit{Seq}_{\mathcal{Q}_i} \otimes \textit{Seq}_{\mathcal{Q}_i}$ denotes the
Kronecker product of sequences, pairing elements from
$\textit{Seq}_{\mathcal{Q}_i}$.
Because $\textit{Seq}_{\mathcal{Q}_i}$ is identical in $M_1$ and $M_2$, it follows that
$\textit{Seq}_{\mathcal{Q}_{i+1}}$ is also identical in both automata,
which establishes
property~(ii).

Consider the down-assignment-language of each state of $M_1$ and $M_2$ at level~$i+1$.
Let $q^{i+1}_{j} \in \textit{Seq}_{\mathcal{Q}_{i+1}}$ be a state at level~$i+1$, and let
\[
\bigl[
(q^{i}_{a_1}, q^{i}_{b_1}) \rightarrow q^{i+1}_j,\,
(q^{i}_{a_2}, q^{i}_{b_2}) \rightarrow q^{i+1}_j,\,
\ldots,\,
(q^{i}_{a_k}, q^{i}_{b_k}) \rightarrow q^{i+1}_j
\bigr]
\]
be the list of transitions leading to $q^{i+1}_j$, where
$q^{i}_{a_1}, \ldots, q^{i}_{a_k}, q^{i}_{b_1}, \ldots, q^{i}_{b_k}
\in \textit{Seq}_{\mathcal{Q}_{i}}$.
Then the down-assignment language of $q^{i+1}_j$ in $M_1$ is given by
\[
L^{\downarrow_{M_1}}(q^{i+1}_j)
=
\bigcup_{k'=1}^{k}
L^{\downarrow_{M_1}}(q^{i}_{a_{k'}})
\;\Vert\;
L^{\downarrow_{M_1}}(q^{i}_{b_{k'}}),
\]
where $\Vert$ denotes the concatenation of words.

Equivalently, for a state $\textit{Seq}_{\mathcal{Q}_{i+1}}[j]$,
the down-assignment language depends only on the down-assignment languages
of the states in $\textit{Seq}_{\mathcal{Q}_{i}}$, and can be written as
\[
L^{\downarrow_{M_1}}(\textit{Seq}_{\mathcal{Q}_{i+1}}[j])
=
\!\!\!\!
\bigcup_{\delta(\textit{Seq}_{\mathcal{Q}_{i}}[k_1],
                 \textit{Seq}_{\mathcal{Q}_{i}}[k_2])
        \rightarrow
        \textit{Seq}_{\mathcal{Q}_{i+1}}[j]}
\!\!\!\!
L^{\downarrow_{M_1}}(\textit{Seq}_{\mathcal{Q}_{i}}[k_1])
\;\Vert\;
L^{\downarrow_{M_1}}(\textit{Seq}_{\mathcal{Q}_{i}}[k_2]).
\]

An analogous equation holds for
$L^{\downarrow_{M_2}}(\textit{Seq}_{\mathcal{Q}_{i+1}}[j])$.
By the induction hypothesis, for all
$0 \leq j \leq |\textit{Seq}_{\mathcal{Q}_{i}}|-1$, the corresponding
down-assignment languages are identical:
\[
L^{\downarrow_{M_1}}(\textit{Seq}_{\mathcal{Q}_{i}}[j])
=
L^{\downarrow_{M_2}}(\textit{Seq}_{\mathcal{Q}_{i}}[j]).
\]
Therefore, for all
$0 \leq j \leq |\textit{Seq}_{\mathcal{Q}_{i+1}}|-1$, we have
\[
L^{\downarrow_{M_1}}(\textit{Seq}_{\mathcal{Q}_{i+1}}[j])
=
L^{\downarrow_{M_2}}(\textit{Seq}_{\mathcal{Q}_{i+1}}[j]),
\]
which establishes
property~(iii).

\paragraph{Top level $l$.}
At the top level~$l$, by the induction hypothesis, the state sequences in $M_1$ and $M_2$
have the same length, the same ordering, and identical down-assignment languages.
At level~$l$, the value function $\mathcal{V}$ induces a bijection from the states to the
domain $\mathcal{D}$.
Because $\textit{Seq}_{\mathcal{Q}_{l}}$ is identical in $M_1$ and $M_2$, the association of each
state $\textit{Seq}_{\mathcal{Q}_{l}}[j]$ with its value $d \in \mathcal{D}$ is the same in both
automata.
Therefore, $M_1$ and $M_2$ have identical value functions at the top level.

Using the above induction,
it follows that the representations of $M_1$ and $M_2$ are identical, and hence
the \TIDD for a pseudo-Boolean function $f$ is a canonical representation of $f$.
\end{proof}

\subsection{Relationship of $L^{\downarrow_f}$ to Proto-CFLOBDDs}
\label{Se:ch6_CFLOBDDRelation}

This section examines the substructures of both \TIDDs and CFLOBDDs from the standpoint of their handling of partitions of the set of strings of length $2^i$.
It shows that (proto-)CFLOBDDs have some additional flexibility that sub-automata of \TIDDs lack.
This understanding becomes important in~\sectrefs{ch6_separation_intuition}{ch6_separation_example},
where we investigate
the differences in the sizes of the \TIDD and CFLOBDD
that represent a given function $f$.

As discussed in~\cite{2211.06818}, a proto-CFLOBDD $C$ at level-$i$ with $k$ exit vertices partitions the space of $2^{2^i}$ strings of length $2^i$, $P(i)$, into $k$ partitions.
Consider
two proto-CFLOBDDs $C$ and $C'$ at level-$i$ that are a part of a bigger CFLOBDD 
$C_g$ that represents a function $g$.
Because of the Contextual-Interpretation Principle, in the context of $C_g$, $C$ and $C'$ can be over the same set of variables, a different set of variables, or both.
In any case, $C$ and $C'$
encode different local sub-functions, and thereby partition $P(i)$ differently.
Moreover,
the partition space of $C$ 
need
not have any correlation with the partition space of $C'$.

In contrast, 
because the transitions of a \TIDD are not specific to the underlying variables, i.e., the same state is reached when \TIDD is run on two trees over different set of variables but with the same assignments, the partition space -- the down-assignment language -- of a \TIDD is
has to be finer than the partition space
of every proto-CFLOBDD at the same level.
That is, in a \TIDD for the function $g$ considered above, there would be a sub-\TIDD that has to be prepared to mimic both $C$ and $C'$.
Consequently, the partition space of the sub-\TIDD has to be finer than both partition spaces -- $C$ and $C'$.

This property can be stated another way.
Let $p_c$ be the number of groupings at level-$l$;
let $\llbracket P_i \rrbracket$ denote the partition space of the $i^{th}$ grouping at level-$l$;
let $k_i$ denote the number of partitions in $\llbracket P_i \rrbracket$;
and let $\llbracket P_i \rrbracket [k]$ denote the $k^{th}$ partition.
If two words $w_1, w_2$ belong to the same $L^{\downarrow_f}(q)$ for some $q \in \mathcal{Q}_l$ of a \TIDD, then they must belong to the same partition in all $\llbracket P_i \rrbracket$ partition spaces and if  two words $w_1, w_2$ do not belong to the same $L^{\downarrow_f}(q)$, then there must exist at least one partition space $\llbracket P_i \rrbracket$ such that $w_1, w_2$ do not belong to the same partition in $\llbracket P_i \rrbracket$.
\[
\forall w_1, w_2 \in L^{\downarrow_f}(q) \text{, for some $q \in \mathcal{Q_i}$}, 
\forall i \in \{1\ldots p_c\}, \exists k \in \{1\ldots k_{i}\}, w_1, w_2 \in \llbracket P_i \rrbracket [k]
\]
and,
\[
\begin{split}
\forall w_1, w_2, (w_1 \in L^{\downarrow_f}(q_1) \neq w_2 \in L^{\downarrow_f}(q_2) \text{, for some $q_1 \neq q_2 \in \mathcal{Q}_i$}) {\implies} \\
\qquad (\exists i \in \{1\ldots p_c\}, \exists {k_1, k_2} \in \{1\ldots k_{i}\}, {k_1 \neq k_2 \land} w_1 \in \llbracket P_i \rrbracket [k_1] \wedge w_2 \in \llbracket P_i \rrbracket [k_2])    
\end{split}
\]

In summary,
the set of states at each level-$l$ of the \TIDD for a function $f$ 
corresponds to the coarsest partitioning that is as fine as
the partition space of each grouping at level-$l$ of the CFLOBDD that represents $f$.

\section{Operations using \TIDDs}
\label{Se:ch6_ops}

In this section, we provide an overview of the algorithms for the operations using \TIDDs.
\TIDDs are represented in memory similar to that of CFLOBDDs.
Every level $i$ of the \TIDD (a state layer) is an object pointing to (1) the state layer at $i-1$, and (2) a list of lists of state numbers (numbered according to a fixed ordering) representing the hyper-edges or the transitions.
If the transitions between two levels -- level-$1$ and level-$2$ -- are
$\delta(q_{10}, q_{10}) \rightarrow q_{20}$, $\delta(q_{10}, q_{11}) \rightarrow q_{20}$, $\delta(q_{11}, q_{10}) \rightarrow q_{21}$, and $\delta(q_{11}, q_{11}) \rightarrow q_{21}$, the transitions are represented as:
$E = [[0,0], [1,1]]$
(a two-dimensional array in row-major order),\footnote{
  Note: We use 0-based array indexing.
}
where $E[0][0] = 0$ represents $\delta(q_{10}, q_{10}) \rightarrow q_{20}$,
$E[0][1] = 0$ represents $\delta(q_{10}, q_{11}) \rightarrow q_{20}$,
$E[1][0] = 1$ represents $\delta(q_{11}, q_{10}) \rightarrow q_{21}$, and
$E[1][1] = 1$ represents $\delta(q_{11}, q_{11}) \rightarrow q_{21}$.
In general, the transitions between two levels -- $i$ and $i+1$ is represented using a 2D array $E$ (a list of lists),
where $E[a][b] = c$ represents the transition $\delta(q_{ia}, q_{ib}) \rightarrow q_{(i+1)c}$.

\begin{figure}[tb!]
    \centering
    \includegraphics[width=0.75\linewidth]{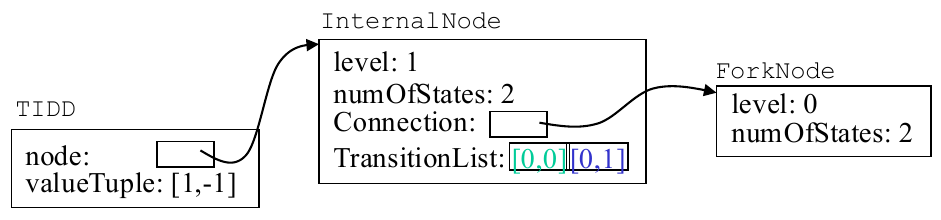}
    \caption{Object-oriented representation of \TIDD for $H_2$.
    }
    \label{Fi:tidd_oops}
\end{figure}

The algorithms for \TIDDs follow a recursive approach — either bottom-up or top-down — and are similar to the algorithms for tree automata, except that they maintain a total order at every level, thereby ensuring canonicity.
\figref{tidd_oops} shows the object-oriented representation of the \TIDD for $H_2$. As discussed in~\exref{ch6_hadamard}, $H_2$ has two states at level-$0$: $q_0, q_1$, and two states at level-$1$: $q_{10}, q_{11}$, with transitions $\delta(q_0, q_0) \rightarrow q_{10}$, $\delta(q_0, q_1) \rightarrow q_{10}$, $\delta(q_1, q_0) \rightarrow q_{10}$, and $\delta(q_1, q_1) \rightarrow q_{11}$. 
The transitions are
represented as a TransitionList: $[[0,0],[0,1]]$.
Every state layer at level-$l$ ($>0$) is represented as an $\tt InternalNode(l)$, and the nodes at level-$0$ are of two kinds: $\tt ForkNode$ (with two states), and $\tt DontCareNode$ (with one state).

\paragraph{Pragmatics.} Techniques such as hash-consing to maintain unique representations, function caching, and equality testing of state layers, and transitions are used and follow a similar pattern to that of CFLOBDDs (and WCFLOBDDs).

\begin{center}
\smallskip
\noindent
{\small{{\tt
\begin{minipage}{\columnwidth}
\begin{tabbing}
Op\=(Node g) \{ \+ \\
    $\ldots$       // Base case: Construct appropriate level-0 node \\
    Internal\=Node g' = new InternalNode(k);  // k = g.level \\
    // Recursive calls on Op(g.Connection) \\
    $\ldots$ \\
    return RepresentativeNode(g'); \- \\
\}
\end{tabbing}
\end{minipage}
}}}
\end{center}

We will define proto-\TIDDs that will be useful in a better understanding of the algorithms.
\begin{defn}
    A proto-\TIDD is similar to a \TIDD with all the constraints except for the final value function $\mathcal{V}$.
    A proto-\TIDD is of the form $M' = (l, \mathcal{Q} = \mathcal{Q}_0 \cup \ldots \cup \mathcal{Q}_l, \mathcal{F} = \{0,1,\Omega\}, \delta)$.
\end{defn}

We will now provide a sketch for algorithms for the operations using \TIDDs.

\subsection{Constant Functions}

\begin{algorithm}[tb!]
\caption{ConstantTIDD \label{Fi:ch6_ConstantTIDD}}
\Input{int l (level), Value $v$,}
\Output{\TIDD representation of a function with $2^l$ variables and constant value $v$}
\Begin{
\Return RepresentativeTIDD(NoDistinctionProtoTIDD(k), [$v$])\;
}
\end{algorithm}

\begin{algorithm}[tb!]
\caption{NoDistinctionProtoTIDD\label{Fi:ch6_NoDistinctionProtoTIDDAlgorithm}}
\Input{int l (level)}
\Output{Proto-TIDD representation of a function with $2^l$ variables}
\Begin{
\If{l == 0}{\Return RepresentativeDontCareNode\;}
InternalNode g = new InternalNode(l)\;
g.numOfStates = 1\;
g.Connection = NoDistinctionProtoTIDD(l-1)\;
g.TransitionList = [[0,0]]\;
\Return RepresentativeNode(g)\;
}
\end{algorithm}

\begin{algorithm}[tb!]
\caption{FalseTIDD\label{Fi:ch6_FalseTIDD}}
\Input{int l (level)}
\Output{\TIDD representation of a function with $2^l$ variables and constant value $F$}
\Begin{
\Return ConstantTIDD(k, $F$);
}
\end{algorithm}

\begin{algorithm}[tb!]
\caption{TrueTIDD\label{Fi:ch6_TrueTIDD}}
\Input{int l (level)}
\Output{\TIDD representation of a function with $2^l$ variables and constant value $T$}
\Begin{
\Return ConstantTIDD(k, $T$);
}
\end{algorithm}

The constant functions $f_0(x) = \lambda x.\bar{0}$ and $f_1(x) = \lambda x.\bar{1}$ have exactly one state at every level.
The \TIDDs for $f_0$ and $f_1$ over $n$ variables have $\log(n)$ levels, with the only difference in $\mathcal{V}$.
The \TIDD representation of $f_0$ and $f_1$, where $n = 2^l$ is a tuple $M = (l = \log(n), \mathcal{Q} = {\mathcal{Q}_0} \cup {\mathcal{Q}_1} \cup \ldots \cup {\mathcal{Q}_l}, \mathcal{Q}_f = {\mathcal{Q}_l}, \mathcal{F} = \{0,1,\Omega\}, \delta, \mathcal{V})$, with
\begin{itemize}
    \item $\mathcal{Q}_0 = \{q_0\}$, $\mathcal{Q}_j = \{q_{j0}\}$, for $j = 1\ldots l$.
    \item $\delta$ function, for $j = 1\ldots l-1$:
    \[
    \begin{array}{cc}
         \delta(\mathbf{0}) \rightarrow q_0 & \delta(\mathbf{1}) \rightarrow q_0\\
         \multicolumn{2}{c}{\delta(q_{0}, q_{0}) = q_{10}}\\
         \multicolumn{2}{c}{\delta(q_{10}, q_{10}) = q_{20}}\\
         \multicolumn{2}{c}{\vdots}\\
         \multicolumn{2}{c}{\delta(q_{j0}, q_{j0}) = q_{(j+1)0}}\\
    \end{array}
    \]
    \item $\mathcal{V}(q_{l0}) = 0$ for $f_0$, and $\mathcal{V}(q_{l0}) = 1$ for $f_1$.
\end{itemize}

The TIDD-creation operations~\algrefsp{ch6_ConstantTIDD}{ch6_FalseTIDD}{ch6_TrueTIDD} internally call
{\tt NoDistinctionProtoTIDD }~\algref{ch6_NoDistinctionProtoTIDDAlgorithm}, which when paired with value tuples $[ 0]$ and $[1]$
(or $[F]$ and $[T]$)
form the \TIDDs for the functions $f_0$ and $f_1$ respectively.

\subsection{Projection Functions}
A second family of creation operations is \emph{(single-variable) projection functions} of the form
$\lambda x_0, x_1, \ldots, x_{n-1} . x_i$, where $i$ ranges from $0$ to $n-1$.

\begin{algorithm}[tb!]
\SetKwFunction{ProjectionTIDD}{ProjectionTIDD}
\SetKwProg{myalg}{Algorithm}{}{end}
  \myalg{\ProjectionTIDD{l, i}}{
  \Input{int k (level), int i (index)}
  \Output{\TIDD representing function $\lambda x_0, x_1, \ldots, x_{n-1} . x_i$}
  \Begin{
    \eIf{i == 0}{
    \Return RepresentativeTIDD(ProjectionProtoTIDD(k,i), [T,F])\;
    }{
    \Return RepresentativeTIDD(ProjectionProtoTIDD(k,i), [F,T])\;
    }
  }
  }{}
  \caption{ProjectionTIDD\label{Fi:ch6_ProjectionTIDDAlgorithm}}
\end{algorithm}

\begin{algorithm}[tb!]
\SetKwFunction{ProjectionProtoTIDD}{ProjectionProtoTIDD}
  \setcounter{AlgoLine}{0}
  \SetKwProg{myproc}{SubRoutine}{}{end}
  \myproc{\ProjectionProtoTIDD{k, i}}{
  \Input{int k (level), int i (index)}
  \Output{Node g representing function $\lambda x_0, x_1, \ldots, x_{2^k-1} . x_i$}
  \Begin{
    \eIf{k == 0}{\Return RepresentativeForkNode\;}
    {
        InternalNode g = new InternalNode(k)\;
        \eIf{i $<$ 2$\ast \ast$(k-1)}{
            g.Connection = ProjectionProtoTIDD(k-1,i)\;
            g.TransitionList = [[0,0],[1,0]]\;
            g.numOfStates = 2\;
        }
        {
            i' = i - 2$\ast \ast$(k-1)\;
            g.Connection = ProjectionProtoTIDD(k-1,i')\;
            g.TransitionList = [[0,1],[0,0]]\;
            g.numOfStates = 2\;
        }
        \Return RepresentativeGrouping(g)\;
    }
  }
  }
  \caption{ProjectionProtoTIDD\label{Fi:ch6_ProjectionTIDDAlgorithmContinued}}
\end{algorithm}

\algrefs{ch6_ProjectionTIDDAlgorithm}{ch6_ProjectionTIDDAlgorithmContinued} show the algorithms for creating Projection functions using \TIDDs.
The level-$0$ node is always $\tt RepresentativeForkNode$
because there need to
be two states at level $0$.
At every level
$k > 0$,
the transition list depends on $i < 2^{k-1}$. If $i \geq 2^{k-1}$, then transition list: $[[0,0],[1,0]]$, else transition list: $[[0,1],[0,0]]$.

\subsection{Unary Operations}
\label{Se:ch6_unaryOps}
\textbf{Scalar Multiplication}. Multiplication of a scalar value $c$ with a \TIDD $M$ 
that represents
function $f$ can be computed by constructing 
the \TIDD $M_c$ that represents the
constant function $\lambda x_0, x_1, \ldots, x_{2^k-1} . c$ (where $2^k$ is the number of variables of $f$), and 
multiplying it
with $M$.
The product
$M \times M_c$ can be computed using the construction provided in~\sectref{ch6_binaryOps}.

\subsection{Binary Operations}
\label{Se:ch6_binaryOps}

The binary operation $\tt op$ applied to two TIDDs $M_1$ and $M_2$, yielding
$M = M_1 \,{\tt op}\, M_2$, is defined via a recursive procedure over the
levels of the input TIDDs. The construction consists of two phases:
(1) a \emph{cross-product (pair-product)} of the states of $M_1$ and $M_2$,
followed by
(2) a \emph{minimization (reduction)} of the TIDD obtained from the
cross-product.

Both the cross-product construction and the subsequent minimization
closely follow the standard product and reduction techniques developed
in the tree-automata literature.

Let
\[
M_1 = (l, \mathcal{Q} = \mathcal{Q}_0 \cup \cdots \cup \mathcal{Q}_l,
\mathcal{F} = \{0,1,\Omega\}, \delta_{\mathcal{Q}}, \mathcal{V}_{\mathcal{Q}})
\]
and
\[
M_2 = (l, \mathcal{P} = \mathcal{P}_0 \cup \cdots \cup \mathcal{P}_l,
\mathcal{F} = \{0,1,\Omega\}, \delta_{\mathcal{P}}, \mathcal{V}_{\mathcal{P}})
\]
be two 
level-$l$
TIDDs.
The \emph{cross-product construction} of $M_1$ and $M_2$ proceeds in a bottom-up
fashion, where states of the resulting TIDD are annotated with pairs of states
from the operand TIDDs as metadata.

Let
\[
M_{PP} = (l, \mathcal{R} = \mathcal{R}_0 \cup \cdots \cup \mathcal{R}_l,
\mathcal{F} = \{0,1,\Omega\}, \delta_{\mathcal{R}}, \mathcal{V}_{\mathcal{R}})
\]
denote the TIDD obtained from the cross-product construction.

\begin{enumerate}
    \item \textbf{Level~0.}
    The set of states $\mathcal{R}_0$ is determined by the Cartesian product of
    $\mathcal{Q}_0$ and $\mathcal{P}_0$, with each state in $\mathcal{R}_0$
    corresponding to a pair of states from the two input TIDDs:

\[
  \begin{array}{ll|rl}
    \multicolumn{1}{c}{\mathcal{Q}_0} & \multicolumn{1}{c}{\mathcal{P}_0} & \multicolumn{2}{c}{\mathcal{R}_0} \\
    \hline
    \{q_0\} & \{p_0\}
         & \{r_0\} &= \{(q_0, p_0)\} \\
    \hline
    \{q_0, q_1\} & \{p_0\}
         & \{r_0, r_1\} &= \{(q_0, p_0), (q_1, p_0)\} \\
    \hline
    \{q_0\} & \{p_0, p_1\}
         & \{r_0, r_1\} &= \{(q_0, p_0), (q_0, p_1)\} \\
    \hline
    \{q_0, q_1\} & \{p_0, p_1\}
         & \{r_0, r_1\} &= \{(q_0, p_0), (q_1, p_1)\} \\
    \hline
  \end{array}
\]

In an object-oriented representation, the above construction can equivalently
be expressed using explicit metadata that records, for each state at a given
level, the corresponding pair of operand states:

\[
  \begin{array}{ll|ll}
    \multicolumn{1}{c}{\mathcal{Q}_0} & \multicolumn{1}{c}{\mathcal{P}_0} &
    \multicolumn{1}{c}{\mathcal{R}_0} & \multicolumn{1}{c}{\text{Metadata}} \\
    \hline
    {\tt DontCareNode} & {\tt DontCareNode}
         & {\tt DontCareNode} & [(0,0)] \\
    \hline
    {\tt ForkNode} & {\tt DontCareNode}
         & {\tt ForkNode} & [(0,0), (1,0)] \\
    \hline
    {\tt DontCareNode} & {\tt ForkNode}
         & {\tt ForkNode} & [(0,0), (0,1)] \\
    \hline
    {\tt ForkNode} & {\tt ForkNode}
         & {\tt ForkNode} & [(0,0), (1,1)] \\
    \hline
  \end{array}
\]

The indices appearing in the metadata correspond to
$\textit{Seq}_{\mathcal{R}_0}$
the canonical numbering of 
level-0 states in $\mathcal{R}$ induced by
$\textit{Seq}_{\mathcal{P}_0}$ and $\textit{Seq}_{\mathcal{Q}_0}$,
which ensures the uniqueness of
the representation
of $\mathcal{R}$.

\item \textbf{Levels $i > 0$.}
Each state at level~$i>0$ of the cross-product TIDD corresponds to a pair of
states from the operand TIDDs. Formally, the set of states at level~$i$ satisfies
\[
\mathcal{R}_i \subseteq \mathcal{Q}_i \times \mathcal{P}_i,
\]
where a state $r \in \mathcal{R}_i$ is represented by the pair
$r = (q, p)$, with $q \in \mathcal{Q}_i$ and $p \in \mathcal{P}_i$.
(Equivalently, this pair is stored as metadata associated with~$r$.)

The transition function of the cross-product TIDD is defined component-wise.
For states $r_a = (q_a, p_a)$ and $r_b = (q_b, p_b)$, let
\[
\delta_{\mathcal{Q}}(q_a, q_b) = q_c
\quad \text{and} \quad
\delta_{\mathcal{P}}(p_a, p_b) = p_c.
\]
Then the transition function $\delta_{\mathcal{R}}$ is given by
\[
\begin{aligned}
\delta_{\mathcal{R}}(r_a, r_b)
  &= \delta_{\mathcal{R}}\bigl((q_a, p_a), (q_b, p_b)\bigr) \\
  &= \bigl(\delta_{\mathcal{Q}}(q_a, q_b), \delta_{\mathcal{P}}(p_a, p_b)\bigr) \\
  &= (q_c, p_c).
\end{aligned}
\]

In an object-oriented implementation, this information is maintained as pairs
of canonical state identifiers at each level, and is used to construct the states
and transitions at the next higher level.

\item \textbf{Level~$l$.}
At the top-most level, the value function of the cross-product TIDD is defined
by combining the values of the corresponding operand states
via operation $\tt op$.
For a state $r = (q, p) \in \mathcal{R}_l$,
\[
\mathcal{V}_{\mathcal{R}}(r)
  = \mathcal{V}_{\mathcal{R}}(q, p)
  = \mathcal{V}_{\mathcal{Q}}(q) \,{\tt op}\, \mathcal{V}_{\mathcal{P}}(p).
\]

\end{enumerate}

The next step consists of \emph{reducing} (or \emph{minimizing}) the TIDD
resulting from the cross-product construction to obtain a minimal,
canonical TIDD.
The minimization procedure follows the standard equivalence-based reduction derived from the Myhill--Nerode theorem,
while also generating, at each level $i$, the canonical naming of states $\textit{Seq}_{\mathcal{R}_i} = [q_0^i, \ldots, q_{|\textit{Seq}_{\mathcal{R}_i}|-1}^i]$.

At the top-most level~$l$, states are merged so as to obtain a one-to-one and onto
correspondence with the output range of the value function
$\mathcal{V}_{\mathcal{R}}$. Let $(r, v)$ denote a state $r$ associated with value
$v$.
For example, suppose that the set of state--value pairs at level~$l$ of
$M_1 \,{\tt op}\, M_2$ is
\[
\{(r_1, v_1), (r_2, v_2), (r_3, v_1)\}.
\]
Since states $r_1$ and $r_3$ are associated with the same value $v_1$, they are
merged, yielding the minimized set of state--value pairs
\[
\{(r_1, v_1), (r_2, v_2)\}.
\]

More precisely, the sequence $Seq_{\mathcal{Q}_i}$ induces a total ordering on the
states at level~$i$. Rewriting the above example in terms of this ordered
representation, if the sequence of state--value pairs at level~$l$ is
\[
[(r_1, v_1), (r_2, v_2), (r_3, v_1)],
\]
then the minimized \TIDD contains the sequence
\[
[(r_1, v_1), (r_2, v_2)],
\]
where state $r_3$ is merged into $r_1$, the leftmost occurrence of a state--value
pair with value $v_1$.

In general, when multiple state--value pairs share the same value $v$, they are
merged into a single representative $(r, v)$, chosen to be the leftmost
occurrence in the ordered sequence. This merging strategy is analogous to the greedy left-folding used in CFLOBDDs.

The equivalence information obtained at level~$l$ is then propagated downward
through the TIDD. State merging is performed iteratively at each lower level,
using the induced equivalence relation on transitions, until level~$0$ is
reached or no further minimization is possible.
As discussed in the case of the top-most level~$l$, the merging of a collection of states $\{ q_i, q_j, q_k\}$, where the state-sequence at a given level is $[\ldots, q_i, \ldots, q_j, \ldots, q_k, \ldots]$, the leftmost instance---here, $q_i$---becomes the representative of the collection.

\subsection{Matrix Multiplication}
Similar to the discussion in~\cite[\S7]{2211.06818}, we use an interleaved variable ordering of the row and column variables to represent matrices.\footnote{In the case of vectors, an ascending or descending order of rows (or columns, based on the representation) is used.}
In this section, we present an algorithm for matrix multiplication using
\TIDDs.
The algorithm proceeds in a recursive, bottom-up manner, incrementally
constructing and propagating information at each level in the form
$(q, p, w)$, where $q$ and $p$ denote states of the operand \TIDDs and
$w$ is the associated weight.
At the top-most level, each tuple $(q, p, w)$ is resolved to a value, after
which the standard reduction procedure is applied to obtain the resulting
\TIDD.

Let
\[
M_1 = (l, \mathcal{Q} = \mathcal{Q}_0 \cup \cdots \cup \mathcal{Q}_l,
\mathcal{F} = \{0,1,\Omega\}, \delta_{\mathcal{Q}}, \mathcal{V}_{\mathcal{Q}})
\]
and
\[
M_2 = (l, \mathcal{P} = \mathcal{P}_0 \cup \cdots \cup \mathcal{P}_l,
\mathcal{F} = \{0,1,\Omega\}, \delta_{\mathcal{P}}, \mathcal{V}_{\mathcal{P}})
\]
be two TIDDs of the same height~$l$.

Let
\[
M_{3} = (l, \mathcal{R} = \mathcal{R}_0 \cup \cdots \cup \mathcal{R}_l,
\mathcal{F} = \{0,1,\Omega\}, \delta_{\mathcal{R}}, \mathcal{V}_{\mathcal{R}})
\]
denote the TIDD obtained from the bottom-up computation of matrix multiplication on $M_1$ and $M_2$.

\begin{enumerate}
    \item \textbf{Level~1.}
    Let us assume that $\mathcal{Q}_0$ and $\mathcal{P}_0$ have two states each at level-$0$.
    If either operand has only a single state, the construction reduces to a trivial special case of the algorithm described below; we therefore omit a detailed discussion of this case. In this setting, the level-0 states of the product TIDD satisfy $\delta_\mathcal{R} (\textbf{0}) \rightarrow r_0$, and $\delta_\mathcal{R} (\textbf{1}) \rightarrow r_1$.
    
    The set of states $\mathcal{R}_1$ is determined as follows:

\[
  \begin{array}{l|r}
    \multicolumn{2}{c}{\mathcal{R}_1}\\
    \hline
    r_{00} & (\delta_\mathcal{Q}(q_0, q_0), \delta_\mathcal{P}(p_0, p_0), 1) + (\delta_\mathcal{Q}(q_0, q_1), \delta_\mathcal{P}(p_1, p_0), 1) \\
    \hline
    r_{01} & (\delta_\mathcal{Q}(q_0, q_0), \delta_\mathcal{P}(p_0, p_1), 1) + (\delta_\mathcal{Q}(q_0, q_1), \delta_\mathcal{P}(p_1, p_1), 1) \\
    \hline
    r_{10} & (\delta_\mathcal{Q}(q_1, q_0), \delta_\mathcal{P}(p_0, p_0), 1) + (\delta_\mathcal{Q}(q_1, q_1), \delta_\mathcal{P}(p_1, p_0), 1) \\
    \hline
    r_{11} & (\delta_\mathcal{Q}(q_1, q_0), \delta_\mathcal{P}(p_0, p_1), 1) + (\delta_\mathcal{Q}(q_1, q_1), \delta_\mathcal{P}(p_1, p_1), 1) \\
    \hline
  \end{array}
\]

Each state in $\mathcal{R}_1$ is represented as a formal sum of weighted triples
$(q, p, w)$, where $q \in \mathcal{Q}_1$, $p \in \mathcal{P}_1$, and
$w \in \mathbb{N}$ is the associated weight.

\medskip
\noindent
\textbf{Semantics (Matrix Multiplication).}
Each triple $(q, p, w)$ corresponds to a contribution of weight $w$ to the
product of the matrix entries represented by states $q$ and $p$.
Specifically, for fixed row index $i$ and column index $j$, the state $r_{ij}$
encodes the summation
\[
r_{ij}
\;\equiv\;
\sum_{k \in \{0,1\}}
\bigl(\delta_{\mathcal{Q}}(q_i, q_k),
      \delta_{\mathcal{P}}(p_k, p_j),
      1\bigr),
\]
which mirrors the standard matrix multiplication rule
\[
(C)_{ij}
\;=\;
\sum_{k} (A)_{ik} \cdot (B)_{kj}.
\]

\medskip
\noindent
\textbf{Addition of triples.}
Given two triples $(q_a, p_a, w_a)$ and $(q_b, p_b, w_b)$, their sum is defined as
\[
(q_a, p_a, w_a) + (q_b, p_b, w_b)
=
\begin{cases}
(q_a, p_a, w_a + w_b),
& \text{if } q_a = q_b \text{ and } p_a = p_b, \\
\text{two distinct triples},
& \text{otherwise}.
\end{cases}
\]
When merging triples with identical state pairs, the representative $(q,p)$ is
chosen according to the fixed total order used for canonicalization.

\medskip
\noindent
Thus, at each level of the construction, the algorithm maintains a unique set of
state pairs $(q,p)$ together with an accumulated weight, exactly capturing the
additive structure of matrix multiplication within the \TIDD framework.

Note that states $r_{00}$, $r_{01}$, $r_{10}$, and $r_{10}$ may be merged whenever their associated sets of triples are identical.

\item \textbf{Levels $i > 0$.}
Each state at level~$i>0$ is represented as a formal sum of weighted triples of
the form
\[
r
=
(q_{j_1}, p_{j_1}, w_{j_1})
+ \cdots +
(q_{j_{k}}, p_{j_{k}}, w_{j_{k}}),
\]
where $q_{j_\ell} \in \mathcal{Q}_i$, $p_{j_\ell} \in \mathcal{P}_i$, and
$w_{j_\ell}$ denotes the associated weight.

Let
\[
\begin{aligned}
r_a &=
(q_{a_1}, p_{a_1}, w_{a_1})
+ \cdots +
(q_{a_{k_1}}, p_{a_{k_1}}, w_{a_{k_1}}), \\
r_b &=
(q_{b_1}, p_{b_1}, w_{b_1})
+ \cdots +
(q_{b_{k_2}}, p_{b_{k_2}}, w_{b_{k_2}})
\end{aligned}
\]
be two such states.
The transition function $\delta_{\mathcal{R}}$ is defined by distributing over
the sums and combining all pairwise contributions:
\[
\delta_{\mathcal{R}}(r_a, r_b)
=
\sum_{u=1}^{k_1}
\sum_{v=1}^{k_2}
\bigl(
\delta_{\mathcal{Q}}(q_{a_u}, q_{b_v}),
\delta_{\mathcal{P}}(p_{a_u}, p_{b_v}),
w_{a_u} \cdot w_{b_v}
\bigr).
\]

As in the level-$1$ construction, the resulting collection of triples is
canonicalized by merging triples with identical state pairs $(q,p)$ and
accumulating their weights. Consequently, at each level, only states with
distinct sets of weighted triples are maintained.

\item \textbf{Level $l$.}
At the top-most level-$l$, each state of
the form
\[
r
=
(q_{j_1}, p_{j_1}, w_{j_1})
+ \cdots +
(q_{j_{k}}, p_{j_{k}}, w_{j_{k}}),
\]

is resolved and mapped to a fixed value as:
\[
\mathcal{V_R}(r) = \mathcal{V_R}((q_{j_1}, p_{j_1}, w_{j_1})
+ \cdots +
(q_{j_{k}}, p_{j_{k}}, w_{j_{k}})) = \sum_{u = 0}^k \mathcal{V_Q}(q_u) \times \mathcal{V_P}(p_u) \times w_u
\]

The unique states and their associated values are computed as shown at the top level. States with the same value are merged---with the left-most instance chosen as the representative---and this information is propagated downward using the
\TIDD reduction
algorithm discussed earlier.

\end{enumerate}

\subsection{Sampling}
\label{Se:ch6_sampling}
A \TIDD with only non-negative values in the range of $\mathcal{V}$ can be considered to represent a discrete distribution over the set of assignments to the Boolean variables.
An assignment---or equivalently, the corresponding binary tree—is considered
to be an elementary event. The probability of the assignment (or tree) 
$p$,
on which \TIDD runs, is the value of $p$ divided by the sum of the values of all assignments (trees) of \TIDD.

We begin by computing
$|L^{\downarrow_f}(q)|$,
the size of $L^{\downarrow_f}(q)$ for every state $q$,where $L^{\downarrow_f}(q)$ denotes the set of Boolean assignments that lead to state $q$.
We also call $|L^{\downarrow_f}(q)|$ the \emph{path count} for $q$.

\textit{Path counting.}
For each state $q$, the
path count
$|L^{\downarrow_f}(q)|$ can be computed recursively in a bottom-up fashion.
Let
\[
\{\delta(q_{a_1}, q_{b_1}) \rightarrow q,\;
  \delta(q_{a_2}, q_{b_2}) \rightarrow q,\;
  \ldots,\;
  \delta(q_{a_k}, q_{b_k}) \rightarrow q\}
\]
be the set of transitions whose target is $q$.
Then,
\[
|L^{\downarrow_f}(q)|
=
\sum_{u=1}^{k}
|L^{\downarrow_f}(q_{a_u})|
\cdot
|L^{\downarrow_f}(q_{b_u})|.
\]

At level~$0$, if the \TIDD contains two states, then the path count of each state
is $1$; otherwise, the unique state has path count $2$.

\textit{Sampling.}
Let $\mathcal{T}$ be a \TIDD
that represents
a (possibly weighted) distribution of the outputs of the Boolean function.
We define a
top-down
procedure to sample an assignment
$a \in \{0,1\}^n$ from the probability distribution induced by $\mathcal{T}$.

\medskip
\noindent
\textbf{Top-level sampling.}
At the top-most level~$l$, let $\mathcal{R}_l$ be the set of states.
Each state $q \in \mathcal{R}_l$ is assigned a weight
\[
W(q) \;\coloneqq\; \mathcal{V}(q) \cdot |L^{\downarrow_f}(q)|,
\]
where $\mathcal{V}(q)$ is the value associated with $q$ and
$|L^{\downarrow_f}(q)|$ denotes the number of assignments leading to $q$.
These weights induce a probability distribution
\[
\Pr[q]
\;=\;
\frac{W(q)}{\sum_{q' \in \mathcal{R}_l} W(q')}.
\]
A state $q$ is sampled according to this distribution.

\medskip
\noindent
\textbf{Recursive step.}
Given a sampled state $q$ at level~$i>0$, let
\[
\mathcal{T}(q)
=
\{|L^{\downarrow_f}(q)|\}
\]
denote the set of transitions whose target is $q$.
One such transition $\delta(q_j, q_k) \rightarrow q$ is selected based on the probability.
\[
  \Pr[\delta(q_j, q_k) \rightarrow q] = \frac{|L^{\downarrow_f}(q_j)| \cdot |L^{\downarrow_f}(q_k)|}{|L^{\downarrow_f}(q)|}.
\]

The sampling procedure is then invoked recursively on states $q_j$ and $q_k$,
yielding assignments $a_j$ and $a_k$, respectively.
The assignment returned for state $q$ is defined as the concatenation
\[
a \;=\; a_j \,\|\, a_k.
\]

\medskip
\noindent
\textbf{Base case (level~$0$).}
At level $0$, each state corresponds to an assignment of a single Boolean
variable.
The sampling behavior depends on the number of paths to a state at level $0$.
If a level-$0$ state $q$ has two transitions corresponding to
  assignments $0$ and $1$, then one of the two assignments is chosen uniformly
  at random. 
  If state $q$ has only one transition, then the assignment corresponding to the transition -- $0$ or $1$ -- is chosen.

In both cases, the base case returns a length-$1$ assignment, which is then
used in the recursive construction of higher-level assignments.

\medskip
\noindent
This recursive procedure samples assignments according to the probability
distribution represented by the \TIDD.

\section{Understanding \TIDDs, CFLOBDDs, and BDDs Through Examples}
\label{Se:ch6_examples}

In this section, we will discuss (1) functions that are efficiently represented by \TIDDs---i.e., have similar efficiency to CFLOBDDs---and can be exponentially smaller than any BDD for the function
(\sectref{ch6_efficient_relations}),
(2) give intuition why CFLOBDDs represent functions better than \TIDDs because of the lack of linear structure in \TIDDs (\sectref{ch6_separation_intuition}),
and (3) discuss an example where there is an exponential gap between CFLOBDDs and \TIDDs in terms of representation size
(\sectref{ch6_separation_example})---i.e., 
the CFLOBDD for a function $f$ can be exponentially smaller than any \TIDD for $f$. 

\subsection{Relations Efficiently Represented by \TIDDs}
\label{Se:ch6_efficient_relations}

In this section, we prove that there exists an exponential separation between \TIDDs and BDDs, 
using functions where the \TIDD is of a similar size to the functions' CFLOBDDs. 
We establish this result using two relations that can be efficiently represented by \TIDDs and CFLOBDDs, but not by BDDs. These examples are taken from~\cite[\S8.1,\S8.2]{2211.06818}, where CFLOBDDs show exponential compression over BDDs.

\subsubsection{The Hadamard Relation $H_n$}
\label{Se:ch6_Separation:HadamardRelation}

We will use the Hadamard Relation $H_n$ to show that \TIDDs efficiently represent (some) functions better than BDDs.

\begin{thm}[Exponential separation for the Hadamard relation]\label{The:ch6_HadamardSeparation}
The Hadamard Relation $H_n:\{0,1\}^{n/2} \times \{0,1\}^{n/2} \rightarrow \{1,-1\}$ between variable sets ($x_0\cdots x_{n/2}$) and ($y_0\cdots y_{n/2}$), where $n = 2^l$, can be represented by a \TIDD of size $\bigO(\log{n})$, similar to that of
the CFLOBDD for $H_n$.
In contrast, 
any
BDD that represents $H_n$ requires $\bigOmega(n)$ nodes.
\end{thm}

\begin{proof}
\smallskip
\noindent
\textit{\TIDD Claim}.
As shown in~\exref{ch6_hadamard}, each matrix $H_n \in \HadamardFamily$, where $n = 2^l$ can be represented by a \TIDD with $\bigO (l)$ vertices and edges---i.e.,
of size
$\bigO (\log n)$.

\smallskip
\noindent
\textit{BDD Claim.}
Regardless of the variable ordering, the BDD representation for $H_n$ requires at least $n$ nodes, one node for each variable in the argument,
as shown in~\cite[\S8.2]{2211.06818}.
\end{proof}

\subsubsection{The Equality Relation $\EQ_n$}
\label{Se:ch6_Separation:EqualityRelation}

We will use the Equality Relation $\EQ_n$ to show that \TIDDs efficiently represent (some) functions better than BDDs.

\begin{defn}\label{De:ch6_EqualityRelation}
The equality relation $\EQ_n : \{0,1\}^{n/2} \times \{0,1\}^{n/2} \rightarrow \{0,1\}$ on variables ($x_0\cdots x_{n/2-1}$) and ($y_0\cdots y_{n/2-1}$) is the relation
$
    \EQ_n(X,Y) \eqdef \Pi_{i=0}^{n/2-1} (x_i \Leftrightarrow y_i)
              =  \Pi_{i=0}^{n/2-1} (\bar{x_i}\bar{y_i} \lor x_i y_i)
$.
\end{defn}

\begin{thm}[Exponential separation for the equality relation]\label{The:ch6_EQSeparation}
For $n = 2^l$, where $l \geq 1$,
$\EQ_n$ can be represented by a \TIDD of size $\bigO(\log{n})$.
In contrast, a BDD that represents $\EQ_n$ requires $\bigOmega(n)$ nodes.
\end{thm}

\begin{proof}
\noindent
\textit{\TIDD Claim.}
We claim that with the interleaved-variable ordering $\langle x_0, y_0, \ldots, x_{n/2-1}, y_{n/2-1} \rangle$, the \TIDD representation of $\EQ_n$ uses $\bigO (\log{n})$ states and edges.

With the interleaved variable ordering, the \TIDD representation of $\EQ_{2^l}$, where $n = 2^l$ is a tuple $M_{\EQ_{2^l}} = (l = l, \mathcal{Q} = {\mathcal{Q}_0} \cup {\mathcal{Q}_1} \cup \ldots \cup {\mathcal{Q}_l}, \mathcal{Q}_f = {\mathcal{Q}_l}, \mathcal{F} = \{0,1,\Omega\}, \delta, \mathcal{V})$, with
\begin{itemize}
    \item $\mathcal{Q}_0 = \{q_0, q_1\}$, $\mathcal{Q}_j = \{q_{j0}, q_{j1}\}$, for $j = 1\ldots l$.
    \item $\delta$ function, for $j = 1\ldots l-1$:
    \[
    \begin{array}{cccc}
         \multicolumn{2}{c}{\delta(\mathbf{0}) \rightarrow q_0} & \multicolumn{2}{c}{\delta(\mathbf{1}) \rightarrow q_1}\\
         \delta(q_{0}, q_{0}) = q_{10} & \delta(q_{0}, q_{1}) = q_{11} & \delta(q_{1}, q_{0}) = q_{11} & \delta(q_{1}, q_{1}) = q_{10}\\
         \delta(q_{10}, q_{10}) = q_{20} & \delta(q_{10}, q_{11}) = q_{21} & \delta(q_{11}, q_{10}) = q_{21} & \delta(q_{11}, q_{11}) = q_{21}\\
         \vdots\\
         \delta(q_{j0}, q_{j0}) = q_{(j+1)0} & \delta(q_{j0}, q_{j1}) = q_{(j+1)1} & \delta(q_{j1}, q_{j0}) = q_{(j+1)1} & \delta(q_{j1}, q_{j1}) = q_{(j+1)1}\\
    \end{array}
    \]
    \item $\mathcal{V}(q_{l0}) = 1$, $\mathcal{V}(q_{l1}) = 0$
\end{itemize}
We observe that $M_{\EQ_{2^i}}$ represents $\EQ_{2^i}$ and has $2i + 2$ states and $4i + 2$ edges.
Therefore, the \TIDD representation of $\EQ_{n}$ is of size $\bigO(\log{n})$.

\smallskip
\noindent
\textit{BDD Claim.}
Regardless of the variable ordering, the BDD representation for $\EQ_n$ requires at least $n$ nodes, one node for each variable in the argument,
as shown in~\cite[\S8.1]{2211.06818}.
\end{proof}

\subsection{Examples Showing the Use of Linearity in CFLOBDDs}
\label{Se:ch6_separation_intuition}

To understand how linearity in CFLOBDDs helps it to represent functions more efficiently than \TIDDs, let us take a look at the example shown in~\figref{ch6_intuition_example}.

\begin{figure}[tb!]
    \centering
    \begin{subfigure}[t]{0.99\linewidth}
        \includegraphics[width=0.9\linewidth]{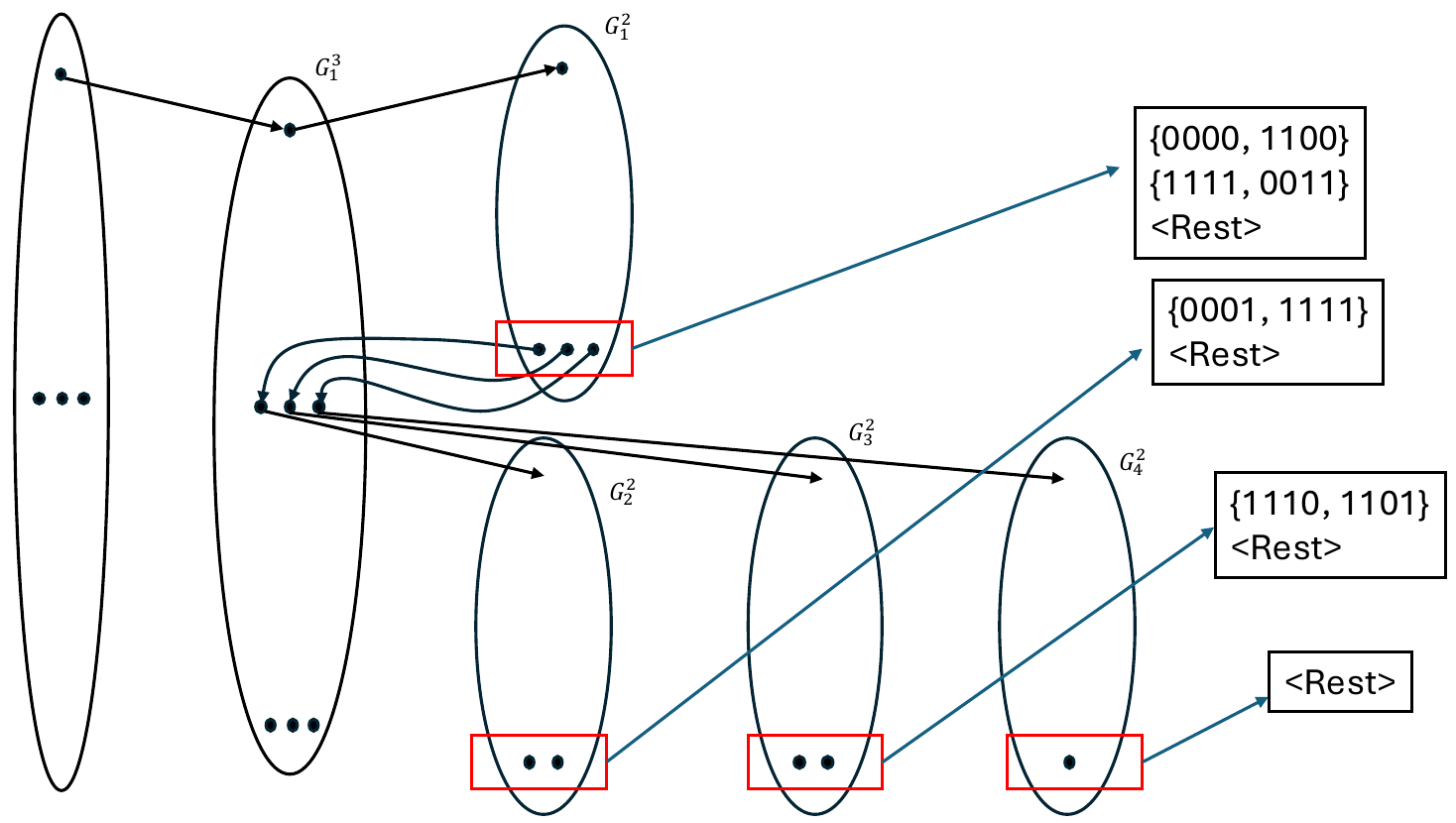}
        \caption{Diagram showing a proto-CFLOBDD with AConnection to a proto-CFLOBDD whose AConnection and BConnections at level-$2$ partition $P(4)$ as shown.}
        \label{Fi:ch6_intuition_example_a}
    \end{subfigure}
    \vspace{1ex}
    \begin{subfigure}[t]{0.99\linewidth}
        \includegraphics[width=0.9\linewidth]{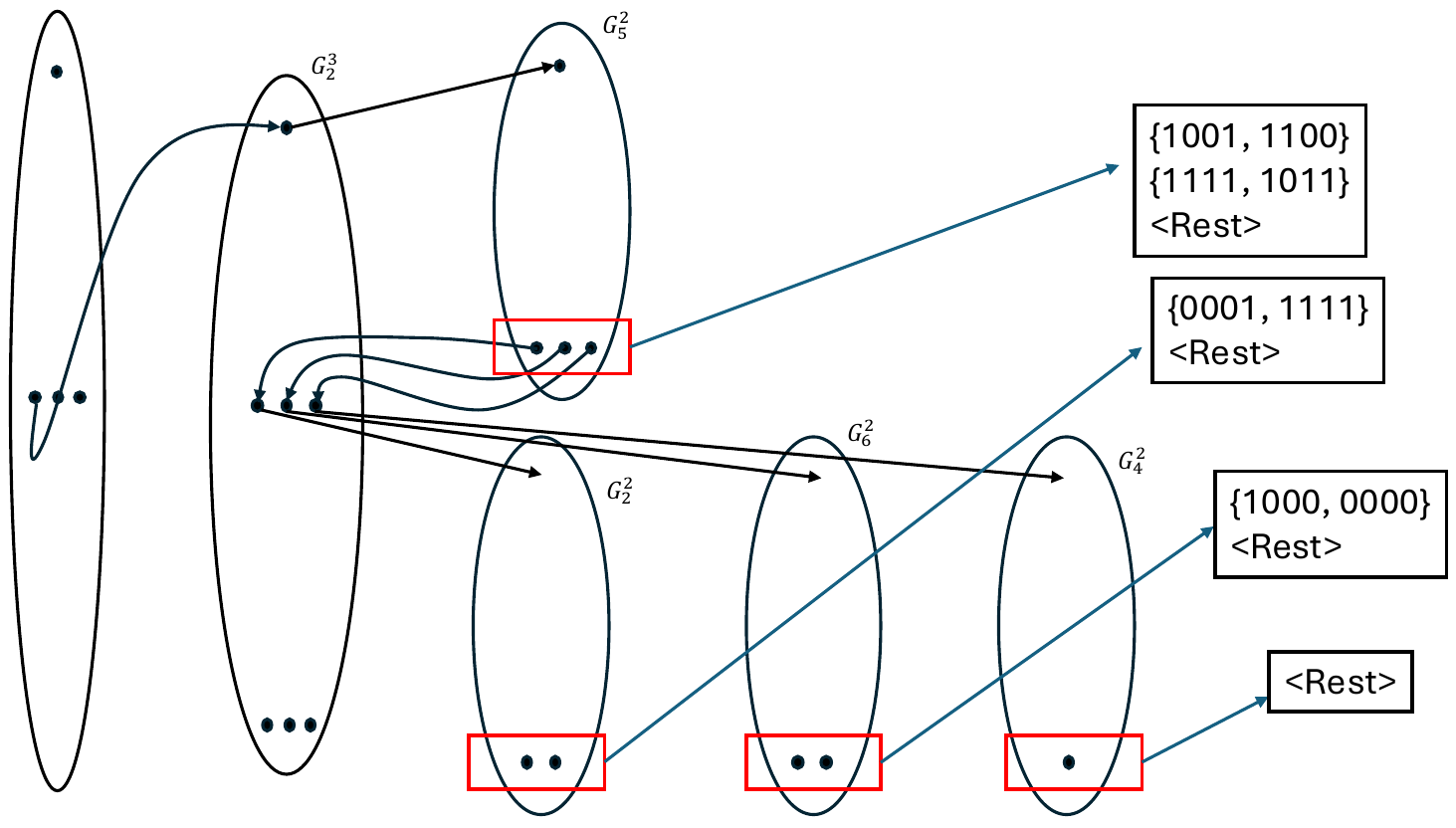}
        \caption{Diagram showing a proto-CFLOBDD with a BConnection to a proto-CFLOBDD whose AConnection and BConnections at level-$2$ partition $P(4)$ as shown.}
        \label{Fi:ch6_intuition_example_b}
    \end{subfigure}
    \caption{
    Diagrams showing a proto-CFLOBDD with an AConnection
    (\figref{ch6_intuition_example_a})
    and a BConnection
    (\figref{ch6_intuition_example_b})
    to proto-CFLOBDDs that partition the space of strings differently; edges not necessary for the discussion are omitted to remove clutter.}
    \label{Fi:ch6_intuition_example}
\end{figure}

\figref{ch6_intuition_example} shows a proto-CFLOBDD that has another proto-CFLOBDD $G_1^3$ at level-$3$ as AConnection and $G_2^3$ at level-$3$ as the first BConnection. Edges not necessary for the discussion are omitted to remove clutter.
Let us understand $G_1^3$ in detail.
$G_1^3$ has an AConnection $G_1^2$ and three BConnections: $G_2^2$, 
$G_3^2$,
and $G_4^2$.
Each $G_x^2$ is a proto-CFLOBDD at level-$2$ over $4$ variables that partition $P(4)$ into distinct partitions:
\[
\begin{array}{rl}
    \llbracket P_{G_1^2} \rrbracket &= [\{0000, 1100\}, \{1111,0011\}, {\tt \langle Rest \rangle}]\\
    \llbracket P_{G_2^2} \rrbracket &= [\{0001, 1111\}, {\tt \langle Rest \rangle}]\\
    \llbracket P_{G_3^2} \rrbracket &= [\{1110, 1101\}, {\tt \langle Rest \rangle}]\\
    \llbracket P_{G_4^2} \rrbracket &= [{\tt \langle Rest \rangle}] = P(4)\\
\end{array}
\]
where ${\tt \langle Rest \rangle}$ means $P(4) \,\backslash\, {\text{all explicitly listed sets}}$.
In the case of $\llbracket P_{G_4^2} \rrbracket$, ${\tt \langle Rest \rangle} = P(4)$.
We observe that there are four unique groupings at level-$2$ and 12 edges between groupings at level-$2$ and level-$3$,
including 5 return edges omitted from the diagram of $G_3^1$.

As discussed in~\sectref{ch6_CFLOBDDRelation}, the states of \TIDD at level-$2$ would depend on the partitions of all the groupings of the CFLOBDD at level-$2$.
For now, let us consider the partitions seen above and deduce the states and $L^{\downarrow_f}$ of the states just based on the partitions seen so far.
Based on the discussion in~\sectref{ch6_CFLOBDDRelation}, the states of the \TIDD at level-$2$ should satisfy the property: If two words belong to the same partition in all the partition spaces of CFLOBDD groupings, they belong to the same state's down-assignment-language in \TIDD.
Similarly, if there exists at least one partition space where two words do not belong to the same partition, then they must belong to the down-assignment-language of different states.
Using this insight, let us list the states at level-$2$ of the \TIDD representing this function and their corresponding $L^{\downarrow_f}$.


\[
\begin{array}{rcl}
     q_1 &:& L^{\downarrow_f}(q_1) = \{0000, 1100\}\\
     q_2 &:& L^{\downarrow_f}(q_2) = \{1111\}\\
     q_3 &:& L^{\downarrow_f}(q_3) = \{0011\}\\
     q_4 &:& L^{\downarrow_f}(q_4) = \{0001\}\\
     q_5 &:& L^{\downarrow_f}(q_5) = \{1110, 1101\}\\
     q_6 &:& L^{\downarrow_f}(q_6) = {\tt \langle Rest \rangle}\\
\end{array}
\]

The above partition satisfies the property discussed.
This leads to a \TIDD having 6 states and 25 edges
at level 2,
which is more than the size in the case of a CFLOBDD representation.
This is because, in CFLOBDDs, the BConnections are constructed based on the exit vertices of the AConnection at a given level, thereby every grouping encodes a local (sub-)function and partitions the space of $P(2^{2^i})$ individually.
As a result, the number of groupings and, particularly,
the number of
edges to represent this information is vastly reduced.

This observation,
that the groupings at a given level $l$ can have different partitions,
provides intuition for why the linear structure inherent in CFLOBDDs (and not inherited by \TIDDs) enables a large function to be decomposed into smaller subfunctions, each of which can be represented efficiently. We now extend this example to illustrate this intuition further.

Let us now consider the proto-CFLOBDD shown in~\figref{ch6_intuition_example_b}.
The same proto-CFLOBDD at level-$4$ has a BConnection to a proto-CFLOBDD $G_2^3$ at level-$3$.
$G_2^3$ has an AConnection to $G_5^2$ and three BConnections: $G_2^2$ (reuse from~\figref{ch6_intuition_example_a}), $G_6^2$, and $G_4^2$ (reuse from~\figref{ch6_intuition_example_a}).
Each of these groupings forms the following partition spaces:

\[
\begin{array}{rl}
    \llbracket P_{G_5^2} \rrbracket &= [\{1001, 1100\}, \{1111, 1011\}, {\tt \langle Rest \rangle}]\\
    \llbracket P_{G_2^2} \rrbracket &= [\{0001, 1111\}, {\tt \langle Rest \rangle}]\\
    \llbracket P_{G_6^2} \rrbracket &= [\{1000, 0000\}, {\tt \langle Rest \rangle}]\\
    \llbracket P_{G_4^2} \rrbracket &= [{\tt \langle Rest \rangle}] = P(4)\\
\end{array}
\]

There are four unique groupings at level-$2$ and 12 edges between groupings at level-$2$ and level-$3$, considering only the ones in~\figref{ch6_intuition_example_b}.
The total number of unique groupings and edges at level-$2$ among AConnection and first BConnection of the top-level CFLOBDD (i.e., considering~\figrefs{ch6_intuition_example_a}{ch6_intuition_example_b}) is 6 groupings at level-$2$ and 24 edges.

The new states of \TIDD at level-$2$ would be:

\[
\begin{array}{rcl}
     q'_1 &:& L^{\downarrow_f}(q'_1) = \{0000\}\\
     q'_2 &:& L^{\downarrow_f}(q'_2) = \{1111\}\\
     q'_3 &:& L^{\downarrow_f}(q'_3) = \{0011\}\\
     q'_4 &:& L^{\downarrow_f}(q'_4) = \{0001\}\\
     q'_5 &:& L^{\downarrow_f}(q'_5) = \{1110, 1101\}\\
     q'_6 &:& L^{\downarrow_f}(q'_6) = {\tt \langle Rest \rangle}\\
     q'_7 &:& L^{\downarrow_f}(q'_7) = \{1100\}\\
     q'_8 &:& L^{\downarrow_f}(q'_8) = \{1001\}\\
     q'_9 &:& L^{\downarrow_f}(q'_9) = \{1011\}\\
     q'_{10} &:& L^{\downarrow_f}(q'_1) = \{1000\}\\
\end{array}
\]

We can observe that with the addition of new partition spaces,
the states of \TIDD increase from 6 to 10.
Particularly, the state $q_1$ bifurcates into $q'_1, {q'_7}$ and $q_6$ bifurcates into $q'_6, q'_8, q'_9, q'_{10}$.
Thereby, the number of states at level-$2$ of the \TIDD is 10, and the number of edges is 81.
This is significantly higher than the
number of groupings and edges in the CFLOBDD.
We also observe that the rate of growth of the states of a \TIDD, and especially the edges, grows quite fast.

This
example
illustrates that the linear structure of CFLOBDDs enables a large function to be decomposed into smaller subfunctions that each admit efficient representations, thereby yielding an efficient representation overall—a decomposition that is not possible in a \TIDD.

\subsection{Exponential Separation between CFLOBDDs and \TIDDs}
\label{Se:ch6_separation_example}

In this section, we present an example showing that CFLOBDDs can be exponentially more succinct than \TIDDs under the same variable ordering. We emphasize that this separation is relatively weak, 
because
it does not rule out the existence of a more favorable variable ordering for \TIDDs. 
However,
the result serves to highlight how the linear structure of CFLOBDDs enables exponentially more succinct representations than \TIDDs.

The function that we use can be thought of as taking a Boolean matrix as input and returning 1 if and only if all elements on the anti-diagonal are 0.
That is, the function returns 1 for a matrix of the form
\[
\begin{bmatrix}
\cdot & \cdot & \cdot & \cdot & 0 \\
\cdot & \cdot & \cdot & 0 & \cdot \\
\cdot & \cdot & 0 & \cdot & \cdot \\
\cdot & 0 & \cdot & \cdot & \cdot \\
0 & \cdot & \cdot & \cdot & \cdot
\end{bmatrix}
\]
As defined in \defref{ch6_ExpoSeparationExample}, the matrix entries are presented in row-major order.

\begin{defn}\label{De:ch6_ExpoSeparationExample}
Function $h_n : \{0,1\}^{n^2} \rightarrow \{0,1\}$ on variables $X = (X_0,\ldots, X_{n-1}) = (x_0\ldots x_{n^2-1})$
-- i.e., where $X_i = (x_{i*n}\ldots x_{i*n+n-1})$ 
is defined as
$
    h_n(X) \eqdef \bigwedge_{i=0}^{n-1} f_i(X_i)
$, with
\[
    f_i(X_i) = 
    \begin{cases}
        1 & \text{$i^{th}$ bit of $X_i$: $x_{i*n + n-1-i} = 0$}\\ 
        0 & \text{otherwise}
    \end{cases}
\]
\end{defn}

Function $f_i$ tests the $i^{\textit{th}}$ row for the pattern ``Is the  $i^{\textit{th}}$ entry of the row equal to 0?''

\begin{thm}[Exponential separation]\label{The:ch6_ExpSeparation}
For $n = 2^l$, where $l \geq 1$,
$h_n$ can be represented by a CFLOBDD
of size
$\bigO(n)$.
In contrast, 
the size of a \TIDD representation of $h_n$ (with the same variable ordering) is $\bigOmega(2^n)$.
\end{thm}

\begin{proof}
\smallskip
\noindent
\textit{CFLOBDD Claim.}
We will use the variable ordering $(x_0\ldots x_{n^2-1})$.
A CFLOBDD representing $h_n$ has $n^2$ variables and $2\log(n) + 1$ levels: $0,\ldots,2\log(n)$.
Every proto-CFLOBDD at level-$\log(n)$ encodes a function over $n$ variables.

The CFLOBDD of interest has groupings of two kinds, stratified by the level at which they appear:
(a) one kind makes up levels $0$ through $\log(n)$;
(b) the other kind makes up levels $\log(n) + 1$ through $2\log(n)$.
The groupings in (a) implement the functions $f_i$, $0 \leq i \leq n-1$;
the groupings in (b) implement the conjunction $\bigwedge_{i=0}^{n-1}$ in the definition of $h_n$.
We start by describing the structure of the groupings in (a).

Let us create a new series of functions $F(2^k,s) = g_{s}(\text{over $2^k$ variables})$, where $k = 1\ldots l$, $2^l = n$, and $s = 0\ldots 2^{k}-1$, and $g_s (\text{over $2^k$ variables}) = 1$, if the $s^{\textit{th}}$ bit of the $2^k$ variables is $0$, otherwise $1$.
We can observe that, for $2^k = n$, $g_s (\text{over $n$ variables}) = f_s(X_s)$.
All $2^{2^k}$ assignments of length $2^k$
can be split into two buckets: $F(2^k,s) = 1$ and $F(2^k,s) = 0$.
The functions in this series over $2^k$ variables would be: $F(2^k, s=0)$, $F(2^k, s=1)$, $\ldots$, $F(2^k, s=2^k-1)$.

\begin{figure}[tb!]
    \centering
    \begin{subfigure}[t]{0.4\linewidth}
        \includegraphics[width=0.5\linewidth]{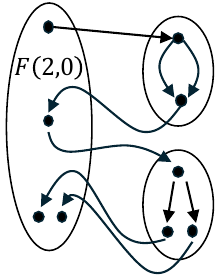}
        \caption{\protect \raggedright proto-CFLOBDD for $F(2,0)$}
        \label{Fi:ch6_f_2_0}
    \end{subfigure}
    \hspace{2ex}
    \begin{subfigure}[t]{0.4\linewidth}
        \includegraphics[width=0.5\linewidth]{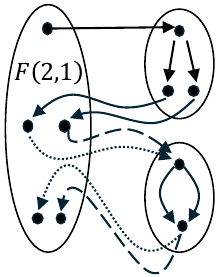}
        \caption{\protect \raggedright proto-CFLOBDD for $F(2,1)$}
        \label{Fi:ch6_f_2_1}
    \end{subfigure}
    \vspace{0.5ex}
    \begin{subfigure}[t]{0.4\linewidth}
        \includegraphics[width=0.6\linewidth]{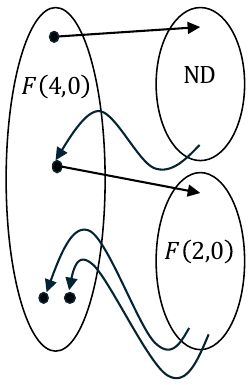}
        \caption{\protect \raggedright proto-CFLOBDD for $F(4,0)$, recursively calls $F(2,0)$}
        \label{Fi:ch6_f_4_0}
    \end{subfigure}
    \hspace{2ex}
    \begin{subfigure}[t]{0.4\linewidth}
        \includegraphics[width=0.6\linewidth]{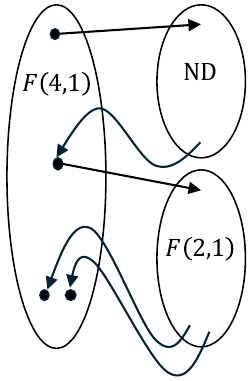}
        \caption{\protect \raggedright proto-CFLOBDD for $F(4,1)$, recursively calls $F(2,1)$}
        \label{Fi:ch6_f_4_1}
    \end{subfigure}
    \vspace{0.5ex}
    \begin{subfigure}[t]{0.4\linewidth}
        \includegraphics[width=0.6\linewidth]{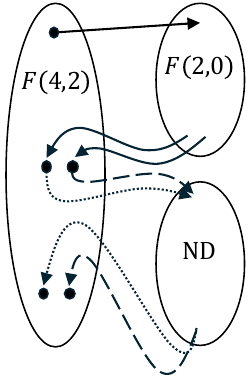}
        \caption{\protect \raggedright proto-CFLOBDD for $F(4,2)$, recursively calls $F(2,0)$}
        \label{Fi:ch6_f_4_2}
    \end{subfigure}
    \hspace{2ex}
    \begin{subfigure}[t]{0.4\linewidth}
        \includegraphics[width=0.6\linewidth]{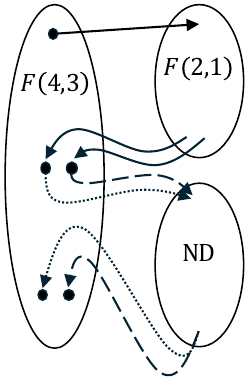}
        \caption{\protect \raggedright proto-CFLOBDD for $F(4,3)$, recursively calls $F(2,1)$}
        \label{Fi:ch6_f_4_3}
    \end{subfigure}
    \caption{\protect \raggedright  The diagrams show proto-CFLOBDDs for $F(2,.)$, $F(4,.)$. $\tt ND$ represents $\tt NoDistinctionNode$.}
    \label{Fi:h_figs}
\end{figure}

\begin{figure}[tb!]
    \centering
    \begin{subfigure}[t]{0.45\linewidth}
        \includegraphics[width=0.6\linewidth]{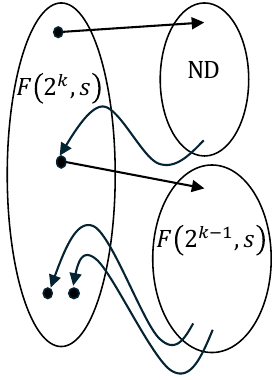}
        \caption{\protect \raggedright proto-CFLOBDD for $F(2^k,s)$, recursively calls $F(2^{k-1},s)$, when $s < 2^{k-1}$}
        \label{Fi:ch6_f_k_s}
    \end{subfigure}
    \begin{subfigure}[t]{0.45\linewidth}
        \includegraphics[width=0.65\linewidth]{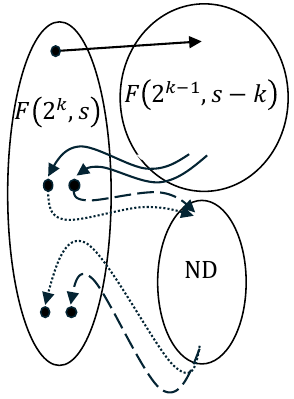}
        \caption{\protect \raggedright proto-CFLOBDD for $F(2^k,s)$, recursively calls $F(2^{k-1},s-k)$, when $s \geq 2^{k-1}$}
        \label{Fi:ch6_f_k_s_2}
    \end{subfigure}
    \caption{\protect \raggedright Diagrams show proto-CFLOBDDs for $F(k,s)$ when $s < 2^{k-1}$ and $s \geq 2^{k-1}$. $\tt ND$ represents $\tt NoDistinctionNode$.}
    \label{Fi:h_figs_general}
\end{figure}

Let us look at some functions in this series:
\begin{enumerate}
    \item $F(2,0)$ splits 
    the assignments of length 2
    into $[\{00,10\}, \{01,11\}]$. 
    The first set consists of the 2-bit strings with a 0 in the $0^{\textit{th}}$ position.
    The proto-CFLOBDD is shown in~\figref{ch6_f_2_0}.
    \item $F(2,1)$ splits 
    the assignments of length 2
    into $[\{00,01\}, \{10,11\}]$.
    The first set consists of the 2-bit strings with a 0 in the $1^{\textit{st}}$ position.
    The proto-CFLOBDD is shown in~\figref{ch6_f_2_1}.
\end{enumerate}

The next functions in this series would be:
\begin{enumerate}
    \item $F(4,0)$ splits
    the assignments of length 4
    into 
    \[
    [\{0000,0010,0100,0110,\ldots,1110\}, \{0001,0011,0101,0111,\ldots,1111\}]
    \]
    The first set consists of the 4-bit strings with a 0 in the $0^{\textit{th}}$ position.
    The proto-CFLOBDD, shown in~\figref{ch6_f_4_0}, recursively calls a $\tt NoDistinctionNode$ for the AConnection and the proto-CFLOBDD for $F(2,1)$ for the BConnection.
    \item $F(4,1)$ splits 
    the assignments of length 4
    into 
    \[
    \left[
    \begin{array}{l}
    \{0000,0001,0100,0101,\ldots,1100,1101\}, \\
    \{0010,0011,0110,0111,\ldots,1110,1111\}
    \end{array}
    \right]
    \]

    The first set consists of the 4-bit strings with a 0 in the $1^{\textit{st}}$ position.
    Similarly, the proto-CFLOBDD, shown in~\figref{ch6_f_4_1}, recursively calls a $\tt NoDistinctionNode$ for the AConnection and the proto-CFLOBDD for $F(2,2)$ for the BConnection.
    \item $F(4,2)$ splits 
    the assignments of length 4
    into 
    \[
    \left[
    \begin{array}{l}
    \{0000,0001,0010,0011,1000,1001,1010,1011\}, \\
    \{0100,0101,0110,0111,1100,1101,1110,1111\}
    \end{array}
    \right]
    \]
    The first set consists of the 4-bit strings with a 0 in the $2^{\textit{nd}}$ position.
    The proto-CFLOBDD, shown in~\figref{ch6_f_4_2}, recursively calls a proto-CFLOBDD for $F(2,1)$ for the AConnection and $\tt NoDistinctionNodes$ for the BConnections.
    \item $F(4,3)$ splits 
    the assignments of length 4
    into 
    \[
    \left[
    \begin{array}{l}
    \{0000,0001,0010,0011,0100,0101,0110,0111\},\\
    \{1000,1001,1010,1011,1100,1101,1110,1111\}
    \end{array}
    \right]
    \]
    The first set consists of the 4-bit strings with a 0 in the $3^{\textit{rd}}$ position.
    The proto-CFLOBDD, shown in~\figref{ch6_f_4_3}, recursively calls a proto-CFLOBDD for $F(2,2)$ for the AConnection and $\tt NoDistinctionNodes$ for the BConnections.
\end{enumerate}

If $s < 2^{k-1}$, then a proto-CFLOBDD for $F(2^k,s)$ has a $\tt NoDistinctionNode$ for AConnection and the proto-CFLOBDD for $F(2^{k-1},s)$ as the BConnection. Similarly, if $s \geq 2^{k-1}$, then the proto-CFLOBDD for $F(2^k,s)$ has the proto-CFLOBDD for $F(2^{k-1},s - k)$ as the AConnection and {\tt NodeDistinctionNodes} as BConnections.

Using this information, we can make the following 
observations about the sizes of these proto-CFLOBDDs.
Let $\mathbb{S}({F(2^k,s)})$ and $\mathbb{E}({F(2^k,s)})$ represent the number of unique groupings and edges of the proto-CFLOBDD representing $F(2^k,s)$,
respectively,
of the proto-CFLOBDD used to represent $F(2^k,s)$.
\[
\begin{array}{lcr|lcr}
     \mathbb{S}({F(2,0)}) & = & 3 & \mathbb{E}({F(2,0)}) & = & 5\\
     \mathbb{S}({F(2,1)}) & = & 3 & \mathbb{E}({F(2,1)}) & = & 7\\
     \mathbb{S}({F(4,0)}) & = & 5 & \mathbb{E}({F(4,0)}) & = & 14\\
     \mathbb{S}({F(4,1)}) & = & 5 & \mathbb{E}({F(4,1)}) & = & 16\\
     \mathbb{S}({F(4,2)}) & = & 5 & \mathbb{E}({F(4,2)}) & = & 16\\
     \mathbb{S}({F(4,3)}) & = & 5 & \mathbb{E}({F(4,3)}) & = & 18\\
\end{array}
\]

For $F(4,0)$, the number of groupings is equal to the number of groupings of $F(2,0)$ + 1 (for $\tt NoDistinctionNode$) + 1 (for grouping at level-$3$) = 5, and the number of edges is the number of edges of $F(2,0)$ + 4 (for $\tt NoDistinctionNode$) + 4 (for edges between level-$2$ and level-$3$ groupings) = 14. Similarly, we can count the number of groupings for others as well.

Extending this
argument to general values of $k$ and $s$,
the number of groupings of $F(k,s)$ is $\mathbb{S}({F(2^k,s)}) = \mathbb{S}({F(2^{k-1},s)}) + 2$, and the number of edges is $\mathbb{E}({F(2^k,s)}) = \mathbb{E}({F(2^{k-1},s)}) + 9$, if $s < 2^{k-1}$ and $\mathbb{S}({F(2^k,s)}) = \mathbb{S}({F(2^{k-1},s-k)}) + 2$ $\mathbb{E}({F(2^k,s)}) = \mathbb{E}({F(2^{k-1},s-k)}) + 11$, otherwise (see~\figrefs{ch6_f_k_s}{ch6_f_k_s_2}).
That is, with each additional increase in $k$, the number of groupings increases by 2, and the number of edges increases by 9-11, i.e., a constant factor.

Now, let us look at the CFLOBDD 
that represents
$h_n$. Each of the groupings at level-$\log (n)$ encodes $f_i$ where $i = 0, \ldots, n-1$.
The proto-CFLOBDD representing $f_i(X_i)$ is encoding the function $F(n,i)$. So, there are $n$ proto-CFLOBDDs corresponding to
the $n$ functions
$f_i(X_i) = F(n,i)$.
Each proto-CFLOBDD recursively calls a proto-CFLOBDD encoding function $F(n/2,i')$, where $i'$ depends on the relationship with $n$.
Because there are $n$ $F(n,i)$ functions, consequently, there would be $n/2$ number of $F(n/2,i')$ functions altogether.
So at every level-$p$, $p = 1, \ldots, \log(n)$, there would be $2^p$ proto-CFLOBDDs representing $F(2^p, j)$, where $j = 0, \ldots, 2^p-1$.
The number of groupings of all proto-CFLOBDDs encoding functions in the series $f_i(X_i) = F(n,i)$,
for $i = 0, \ldots, n-1$,
is $n + n/2 + n/4 + \ldots + 1 = 2n$.
The number of edges is constant per grouping and hence linear in the number of groupings.
The total size of all proto-CFLOBDDs encoding functions in the series $f_i(X_i) = F(n,i)$ is $\bigO(n)$.

\begin{figure}[tb!]
    \centering
    \begin{subfigure}[t]{0.45\linewidth}
        \includegraphics[width=\linewidth]{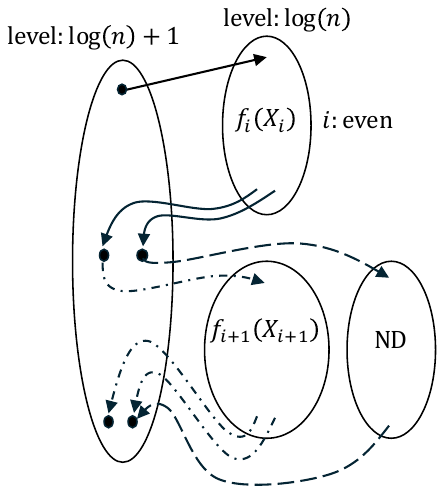}
        \caption{\protect \raggedright The diagram shows the proto-CFLOBDD structure at level-$\log(n) + 1$. The groupings at level-$\log(n) + 1$ have an AConnection to $f_i(X_i)$ (where $i$ is even), and two BConnections: $f_{i+1}(X_{i+1})$ with return tuple $[1,2]$ and $\tt NoDistinctionNode$ with return tuple $[2]$.}
        \label{Fi:ch6_h_logn}
    \end{subfigure}
    \hspace{2ex}
    \begin{subfigure}[t]{0.43\linewidth}
        \includegraphics[width=\linewidth]{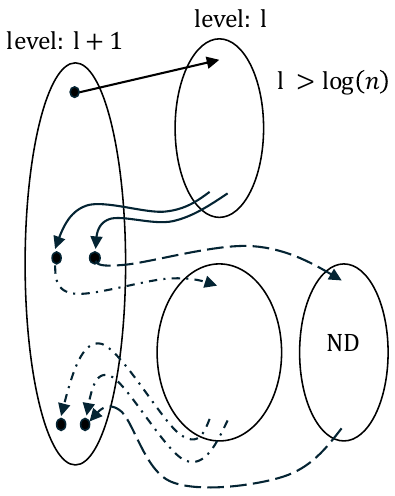}
        \caption{\protect \raggedright The diagram shows the proto-CFLOBDD structure at level-$l+1$, $l > \log(n)$. The groupings at level-$l + 1$ have an AConnection and two BConnections as shown.}
        \label{Fi:ch6_h_l}
    \end{subfigure}
    \caption{Diagrams showing proto-CFLOBDDs of $h_n$ at level-$\log(n)+1$ (\figref{ch6_h_logn}) and level-$l+1$, $l > \log(n) + 1$(\figref{ch6_h_l}).}
    \label{fig:placeholder}
\end{figure}

Note that we still need to count the number of groupings and edges
at levels $\log(n) + 1$ through $2\log(n)$.
These groupings will have the structure of a perfect binary tree of height $\log(n)$.
At every level-$l$($> \log(n)$), the proto-CFLOBDD (see~\figref{ch6_h_logn}) would have two exits with (1) a recursive call to a unique AConnection grouping with two return edges,
(2) a unique grouping through the first BConnection with two return edges to exits: $[1,2]$, (3) a $\tt NoDistinctionNode$ through the second BConnection with one return edge to exit: $[2]$.
At every level-$l$, $l = \log(n) + l'$,
$0 < l' \leq \log(n)$
(see~\figref{ch6_h_l}), there would be $n/2^{l'}$ groupings and a constant number of edges per grouping.
So the total
number
of groupings and edges from level-$\log(n) + 1$ to level-$2\log(n)$ = $n/2 + n/4 + \ldots + 1 = n-1 = \bigO(n)$.
(We used only the
count of the number of
groupings in the calculation, because the number of edges is constant per grouping.)

The total size of
the CFLOBDD that represents
$h_n$ is $\bigO(n)$ + $\bigO(n)$ = $\bigO(n)$.

\smallskip
\noindent
\textit{\TIDD Claim.}

Let us focus on the level-$\log(n)$ of the \TIDD representation for $h_n$.
The number of states at level-$\log(n)$ would be the partitions formed by $f_i(X_i)$.
Each $f_i(X_i)$ creates a partition space over $2^n$ assignments, and the number of 
states is equal to the distinct partitions of assignments satisfying all $f_i(X_i)$ partition spaces.
We will calculate this by considering each $f_i(X_i)$ at a time.

\begin{itemize}
    \item $i = 0$: $f_0$ divides the assignments into $[\{0,2,4,6,\ldots\}, \{1,3,5,7,\ldots\}]$ (we are using the integer values of the Boolean assignments). \#states of \TIDD = 2.
    \item $i = 1$: $f_1$ divides the assignments into $[\{0,1,4,5,\ldots\}, \{2,3,6,7,\ldots\}]$. 
    The new states would have the following $L^{\downarrow_f}$ (again, using integer values):
    $[\{0,4,8,\ldots\}, \{1,5,9,\ldots\}, \{2,6,10,\ldots\},  \{3,7,11,\ldots\}]$. \#states = 4.
    \item $i = 2$: $f_2$ divides the assignments into $[\{0,1,2,3,\ldots\}, \{4,5,6,7,\ldots\}]$. 
    The new states would have the following $L^{\downarrow_f}$ (again, using integer values):
    $[\{0,8,\ldots\}, \{1,9,\ldots\}, \{2,10,\ldots\},  \{3,11,\ldots\}, \ldots, \{7,15,\ldots\}]$. \#states = 8.
    \item $\ldots$
    \item $i = s$: the new states would partition the space of $2^n$ strings into $2^{s+1}$ partitions,
    where the $j^{th}$ partition has elements that satisfy $e \% 2^{s+1} = j$, where $e$ is an element in the $j^{th}$ partition.
    Hence, \#states = $2^{s+1}$.
\end{itemize}

Hence, for $i = n-1$, i.e., the total number of states at level-$\log n$ of
the \TIDD, considering all $f_i$ functions,
is $2^{n+1} = \bigOmega(2^n)$.
Therefore, the total number of states of
the \TIDD that represents
$h_n$ is $\bigOmega(2^n)$.

\end{proof}

This example illustrates how the linear structure of CFLOBDDs enables an exponential compression compared to \TIDDs under the same variable ordering. In a CFLOBDD, each subfunction \(f_i(X_i)\) can be represented independently by a proto-CFLOBDD, and the proto-CFLOBDD for \(f_{i+1}(X_{i+1})\)
is connected directly to the output of \(f_i(X_i)\)---see \figref{ch6_h_logn}.
In contrast, \TIDDs lack this linear compositionality: the same set of states must be reused to represent all subfunctions \(f_{i+1}(X_{i+1})\), which forces the representation to encode exponentially many distinctions and thus leads to an exponential blow-up in the number of states.
\section{Evaluation}
\label{Se:ch6_eval}

In this section, we provide some empirical evidence of how CFLOBDDs perform better than \TIDDs in a practical domain,
thereby showing that the linear structure in CFLOBDDs helps in more efficient representations than \TIDDs
in case of representing Boolean functions.
We use the domain of quantum 
simulation---in particular, the 
Greenberger–Horne–Zeilinger
(\GHZ), Bernstein–Vazirani
(\BV), and Deutsch–Jozsa
(\textit{DJ}) algorithms where CFLOBDDs were shown to perform exceptionally
better than BDDs.
We ran our experiments on
on a system running Ubuntu with an
Intel(R) Xeon(R) Gold 5218 CPU @ 2.30GHz, 64 vCPUs, and 16GB of memory.
Each experiment was executed five times for increasing numbers of qubits, and we report the average results across these runs.

\begin{table}[tb!]
    \centering
    \begin{adjustbox}{width=0.99\linewidth}
    \begin{tabular}{|c|c|c|c|c|c|c|c|c|c|c|c|}
    \hline
         \multirow{2}{*}{Benchmark} & \multirow{2}{*}{\#Qubits} & \multicolumn{5}{c|}{CFLOBDD} & \multicolumn{5}{c|}{\TIDD} \\
         \cline{3-12}
         & & \#Groupings & \#Edges & \#Total & Time(s) & Max. Size & \#Nodes & \#Edges & \#Total & Time(s) & Max. Size\\
         \hline
         \multirow{6}{*}{\GHZ} & 131072 & 139 & 831 & 970 & \textbf{14.87} & 970 & 19 & 213 & \textbf{232}	& 35.22 &	\textbf{233} \\
         \cline{2-12}
         & 262144 & 147 & 879 & 1026 & \textbf{36.21} & 1026 & 20 & 225 & \textbf{245} & 71.16 & \textbf{246}\\
         \cline{2-12}
         & 524288 & 155 & 927 & 1082 & \textbf{97.25} & 1082 & 21 & 237 & \textbf{258} & 144 & \textbf{259}\\
         \cline{2-12}
         & 1048576 & 163 & 975 & 1138 & \textbf{228.21} & 1138 & 22 & 249 & \textbf{271} & 292.9 & \textbf{272}\\
         \cline{2-12}
         & 2097152 & \multicolumn{5}{c|}{Timeout (15min)} & 23 & 261 & \textbf{284} & \textbf{609.81} & \textbf{285}\\
         \cline{2-12}
         & 4194304 & \multicolumn{5}{c|}{Timeout (15min)} & \multicolumn{5}{c|}{Timeout (15min)}\\
         \hline
         \multirow{4}{*}{\BV} & 16 & 30 & 174 & 204 & \textbf{0.003} & \textbf{242} & 7 & 143 & \textbf{150} & 0.012 & 321\\
         \cline{2-12}
         & 32 & 40 & 237 & 277 & \textbf{0.004} & \textbf{324} & 8 & 240 & \textbf{248} & 0.403 & 832\\
         \cline{2-12}
         & 64 & 54 & 323 & \textbf{377} & \textbf{0.005} & \textbf{434} & 9 & 449 & 458 & 200.12 & 37731\\
         \cline{2-12}
         & 128 & 77 & 456 & \textbf{533} & \textbf{0.007} & \textbf{617} & \multicolumn{4}{c|}{Timeout (15min)} & 134480\\
         \hline
         \multirow{5}{*}{\textit{DJ}} & 512 & 33 & 170 & \textbf{203} & \textbf{0.004} & \textbf{290} & 12 & 202 & 214 & 1.27 & 21025\\
         \cline{2-12}
         & 1024 & 36 & 186 & \textbf{222} & \textbf{0.006} & \textbf{318} & 13 & 222 & 235 & 10.37 & 65924\\
         \cline{2-12}
         & 2048 & 39 & 202 & \textbf{241} & \textbf{0.009} & \textbf{346} & 14 & 242 & 256 & 75.46 & 200207\\
         \cline{2-12}
         & 4096 & 42 & 218 & \textbf{260} & \textbf{0.012} & \textbf{374} & 15 & 262 & 277 & 659.62 & 570940\\
         \cline{2-12}
         & 8192 & 45 & 234 & \textbf{279} & \textbf{0.024} & \textbf{402} & \multicolumn{5}{c|}{Timeout (15min)}\\
         \hline
    \end{tabular}
    \end{adjustbox}
    \caption{Performance of CFLOBDDs and \TIDDs on quantum benchmarks---\GHZ, \BV, \textit{DJ}.
    }
    \label{Ta:ch6_quantum_exp}
\end{table}

\tableref{ch6_quantum_exp} reports the performance of CFLOBDDs and \TIDDs on three quantum benchmarks---\GHZ, \BV, and \textit{DJ}, chosen to illustrate the effect of the absence of linearity in \TIDDs. We measure both the construction time and the size of the data structures representing the final quantum state for each benchmark.
\tableref{ch6_quantum_exp} also shows the maximum size of the intermediate state vectors or
(unitary) matrices representation gates
created when running the benchmark, which has a huge impact on the overall performance of \TIDDs and CFLOBDDs.

The size of a \TIDD is computed as the sum of (i) the number of internal and leaf nodes (illustrated in \figref{tidd_oops}) and (ii) the number of edges in the transition lists. We note that the states are not explicitly stored; consequently, there is exactly one node per layer, and the total number of nodes equals the number of levels in the \TIDD.
Furthermore, each transition list is represented as a list of lists. These lists are hash-consed, and each distinct list is counted only once when computing the overall size.

For the \GHZ~benchmark, as shown in~\tableref{ch6_quantum_exp}, we observe that \TIDDs and CFLOBDDs perform comparably, with \TIDDs exhibiting slightly better performance. This is attributable to the structure of the benchmark, which consists of a sequence of simple matrix-multiplication operations, resulting in relatively small intermediate states and gates.
In contrast, for the \BV and \textit{DJ} benchmarks, CFLOBDDs significantly outperform \TIDDs in both time and space. Although the sizes of the final state vectors represented by CFLOBDDs and \TIDDs are similar, the maximum sizes of the intermediate states and gates\footnote{
  In the case of the \BV benchmark for 128 qubits, we report the maximum size observed so far before \TIDDs timeout.
}
differ substantially. In particular, for \TIDDs, the intermediate representations can grow to as many as 570K nodes and edges, whereas CFLOBDDs can represent both intermediate states and gates much more compactly, which directly affects the time taken to run the benchmark.

This behavior arises because CFLOBDDs can decompose a function into smaller sub-functions, represent each sub-function efficiently, and then compose them using the inherent linear structure of CFLOBDDs---a capability that is not available in \TIDDs. Consequently, these results demonstrate that, even in practice, the linear structure of CFLOBDDs plays a crucial role in enabling significantly more compact function representations.
\section{Related Work}
\label{Se:related}

Over the years, BDDs and their variants---including Multi-Terminal BDDs~\cite{dac:CMZFY93,CMU:CS-95-160}, 
Algebraic Decision Diagrams~\cite{DBLP:journals/fmsd/BaharFGHMPS97}, Free Binary Decision Diagrams 
(FBDDs)~\cite[\S6]{DBLP:books/siam/Wegener00}, Binary Moment Diagrams (BMDs)~\cite{dac:BC95}, Hybrid 
Decision Diagrams (HDDs)~\cite{CMU:CS-95-160}, Differential BDDs~\cite{LNCS:AMU95}, Indexed BDDs 
(IBDDs)~\cite{toc:JBAAF97}, and Weighted BDDs (WBDDs)~\cite{Book:ZW2020}---have been employed for 
efficient function representation, with applications spanning program analysis, model checking, and 
quantum simulation.
More recently, CFLOBDDs and their weighted variants 
(WCFLOBDDs)~\cite{DBLP:journals/corr/abs-2305-13610} have demonstrated better compression 
than
BDDs and WBDDs, and have achieved promising results in the domain of quantum simulation~\cite{CAV:SCR23}.
This paper focuses on CFLOBDDs and investigates whether their linear structure is essential to their representational efficiency.
Because
the hierarchical and linear properties of CFLOBDDs are preserved in their weighted variants, the results presented here extend naturally to WCFLOBDDs as well.

Other data structures that generalize BDDs include Sentential Decision Diagrams 
(SDDs)~\cite{darwiche2011sdd} and Variable Shift SDDs~\cite{nakamura2020variable}. These data 
structures extend BDDs by imposing a tree-shaped ordering over variables, and there exist functions 
for which they achieve double-exponential compression over decision trees and exponential compression 
over BDDs. However, SDDs cannot be characterized as deterministic tree automata, and therefore fall 
outside the scope of the comparison with CFLOBDDs in the context of linearity discussed in this paper.

Another approach to introducing hierarchy into BDDs is the notion of \emph{Tree BDDs}~\cite{mcmillan1994hierarchical}. Tree BDDs employ a hierarchical representation that closely resembles a tree automaton. In this framework, the variable set of a Boolean function is recursively partitioned into three components: root variables, left-child variables, and right-child variables. The left and right partitions recursively represent their corresponding sub-functions using the same construction, while the root partition combines the equivalence classes induced by the left and right sub-functions to define the overall function.
A key distinction from CFLOBDDs is that, in Tree BDDs, each partition explicitly associates variables with the corresponding sub-functions. In contrast, CFLOBDDs reuse substructures, wherein groupings are not explicitly tied to particular variables. Even if the Tree BDD framework were modified to eliminate explicit variable instantiation, the resulting structures would still be closely aligned with tree automata and could therefore encounter challenges similar to those faced by \TIDDs, due to the absence of a linear structure.

\section{Conclusion}
\label{Se:conclusion}

This paper investigated whether the linear structure present in CFLOBDDs, in addition to hierarchy, is essential to their succinctness, or whether CFLOBDDs are instead closer in power to a restricted class of tree automata that exploit hierarchy alone.
To answer this question, we introduced a new data structure---Tree‑Automata‑Inspired Decision Diagrams (\TIDDs)---and compared the expressive power of CFLOBDDs and \TIDDs.

We observed that although \TIDDs, like CFLOBDDs, exhibit exponential separation over BDDs for certain families of functions, this behavior does not extend to the general case. This limitation arises because the linear structure of CFLOBDDs enables a disciplined decomposition of a function into subfunctions, followed by their systematic composition according to a prescribed linear ordering. Each subfunction can be represented compactly in isolation, and this ordered composition ensures that these compact representations are reused effectively, yielding an efficient representation of the overall function. In contrast, \TIDDs lack this mechanism, which restricts their ability to generalize the compression advantages of CFLOBDDs.

We formally proved that there exists a family of functions for which, under the same variable ordering, CFLOBDDs are exponentially more succinct than \TIDDs. To assess whether this theoretical gap translates to practical differences, we evaluated both data structures in the domain of quantum-circuit simulation. Our experiments showed that \TIDDs performed significantly worse than CFLOBDDs, primarily due to substantial blow-ups in the representation sizes of intermediate state vectors and unitary matrices. Taken together, these findings confirm that the linear structure in CFLOBDDs is as important as the hierarchy for efficient Boolean functions representations.

\bibliographystyle{ACM-Reference-Format}
\bibliography{refs, df, logic, mab}

\end{document}